\documentclass[11pt,a4paper]{amsart}
\usepackage[foot]{amsaddr}

\usepackage{ifxetex}
\ifxetex
  \usepackage[no-math]{fontspec}
  \usepackage[libertine]{newtxmath}   
\else
  \usepackage[tt=false]{libertine}     
  \usepackage{newtxmath} 
\fi

\usepackage{xcolor}
\usepackage{fullpage}
\usepackage{microtype}
\usepackage{footnote}
\usepackage[linesnumbered,boxed,ruled,vlined]{algorithm2e}
\usepackage{algpseudocode}

\usepackage{array}                  
\usepackage{tabularx}               
\usepackage{longtable}              
\usepackage{multirow}               
\usepackage{makecell}               
\usepackage{tablefootnote}          
\makesavenoteenv{tabular}           

\usepackage[authoryear]{natbib}
\usepackage{hyperref}
\usepackage{cleveref}               

\hypersetup{colorlinks=true,citecolor=blue, linkcolor=blue, urlcolor=blue}

\usepackage{tcolorbox}
\usepackage{enumerate} 
\usepackage{enumitem}
\usepackage{float}
\usepackage[textsize=tiny]{todonotes}

\newtheorem{theorem}{Theorem}[section]

\newtheorem*{claim*}{Claim}
\newtheorem{condition}[theorem]{Condition}
\newtheorem{example}[theorem]{Example}

\newtheorem{lemma}[theorem]{Lemma}

\newtheorem{problem}[theorem]{Problem}
\newtheorem{definition}[theorem]{Definition}
\newtheorem{remark}[theorem]{Remark}
\newtheorem*{remark*}{Remark}

\renewcommand{\emptyset}{\varnothing}
\renewcommand{\epsilon}{\varepsilon}

\newcommand{\E}[1]{\mathbb{E}\left[{#1}\right]}
\renewcommand{\E}[2][]{ \ifthenelse{\isempty{#1}}
  {\mathop{\mathbb{E}}\left[#2\right]} {\mathop{\mathbb{E}}_{#1}\left[#2\right]} }
\renewcommand{\Pr}[2][]{ \ifthenelse{\isempty{#1}}
  {\mathop{\mathbf{Pr}}\left[#2\right]} {\mathop{\mathbf{Pr}}_{#1}\left[#2\right]} }

\newcommand{\abs}[1]{\left\vert#1\right\vert}
\newcommand{\set}[1]{\left\{#1\right\}}

\newcommand{\defeq}{:=}

\newcommand{\DTV}[2]{d_{\mathrm{TV}}\left({#1},{#2}\right)}

\newcommand{\bl}{\mathcal{L}}
\newcommand{\slv}{\mathcal{S}}
\newcommand{\tlv}{\mathcal{T}}
\newcommand{\dlv}{\mathcal{D}}
\newcommand{\qlv}{\mathcal{Q}}

\newcommand{\conf}{\omega}

\newcommand{\uniform}{\mathop{\mathrm{Uniform}}}
\algnewcommand{\IfThen}[2]{\textbf{if}~#1~\textbf{then}~#2}

\newcommand{\update}{\textsc{Update}}
\newcommand{\seeding}{\textsc{Seeding}}
\newcommand{\compress}{\textsc{Compress}}
\newcommand{\disjoint}{\textsc{Disjoint}}

\title{Tight Bounds for Sampling $q$-Colorings via Coupling from the Past}
\author{Tianxing Ding \quad Hongyang Liu \quad Yitong Yin \quad Can Zhou}

\thanks{State Key Laboratory for Novel Software Technology, New Cornerstone Science Laboratory, Nanjing University, 163 Xianlin Avenue, Nanjing, Jiangsu Province, China. Emails: \url{652024330006@smail.nju.edu.cn},\url{liuhongyang@smail.nju.edu.cn},\url{yinyt@nju.edu.cn},\url{502024330075@smail.nju.edu.cn}}

\begin{document}

\begin{abstract}
The Coupling from the Past (CFTP) paradigm is a canonical method for perfect Sampling.
For uniform sampling of proper $q$-colorings in graphs with maximum degree $\Delta$, the bounding chains of [Huber, STOC’98] provide a systematic framework for efficiently implementing CFTP algorithms within the classical regime $q\ge (1+o(1))\Delta^2$.
This was subsequently improved to $q > 3\Delta$ by [Bhandari and Chakraborty, STOC’20] and to $q \ge (\tfrac{8}{3} + o(1))\Delta$ by [Jain, Sah, and Sawhney, STOC’21].

In this work, we establish the asymptotically tight threshold for bounding-chain–based CFTP algorithms for graph colorings.
We prove a lower bound showing that all such algorithms satisfying the standard contraction property require $q \ge 2.5\Delta$, 
and we present an efficient CFTP algorithm that achieves this asymptotically optimal threshold $q \ge (2.5 + o(1))\Delta$ via an optimal design of bounding chains.
\end{abstract}

\maketitle
\vspace{-1cm}
\section{Introduction} 
\label{sec:intro}

Uniform sampling of graph colorings is a fundamental problem that has attracted extensive attention in combinatorics, probability theory, and theoretical computer science.
Given a graph $G = (V, E)$ and an integer $q \geq 1$, a \emph{proper $q$-coloring} (or simply a \emph{$q$-coloring}) is a mapping $\conf: V \to [q]$ such that $\conf(u) \neq \conf(v)$ for every edge $(u, v) \in E$. 
The set of all proper $q$-colorings of $G$ is 
$$
\Omega = \set{\conf: V \to [q] \mid \forall (u,v) \in E,\ \conf(u) \neq \conf(v)}.
$$
Let $\Delta = \Delta(G)$ denote the maximum degree of $G$. 
It is well known that the condition $q \geq \Delta + 1$ guarantees the existence of a proper $q$-coloring, i.e., $\Omega\neq\emptyset$.


A canonical approach to sampling from $\Omega$ is the \emph{Markov chain Monte Carlo} (MCMC) method, 
in which a Markov chain $(X_t)$ is designed to mix rapidly to the uniform distribution $\pi$ over $\Omega$.
Starting from an arbitrary initial state $X_0\in\Omega$, simulating the chain for a sufficiently long time $T$ yields a sample $X_T$ approximately distributed as $\pi$.
A well studied chain is the \emph{Glauber dynamics}, which at each step selects a vertex uniformly at random and recolors it with a color not currently used by its neighbors.
When $q \ge \Delta+2$, the chain is ergodic and converges to the uniform distribution $\pi$ over $\Omega$.
Determining the critical threshold on $q$ for rapid mixing is a long-standing open problem.
%
%
Over the past decades, a sequence of works has progressively improved this threshold, 
from the classical bound $q > 2\Delta$~\cite{jerrum1995very,salas1997absence,bubley1997path} 
to the current best bound $q > 1.809\Delta$~\cite{vigoda1999improved,chen2019improved,carlson2025flip}.

\subsection{Coupling from the past and grand coupling}
\label{subsec:cftp-grand-coupling}

The \emph{coupling from the past} (CFTP) technique, introduced by Propp and Wilson~\cite{propp1996exact}, provides a general framework for the \emph{perfect simulation} of the stationary distribution of a Markov chain.
%
The key idea is to consider an idealized Markov chain that has been running from time $-\infty$ up to time $0$, so that its final state $X_0\sim\pi$ is an exact (perfect) sample from the stationary distribution~$\pi$.
To generate such a sample, the CFTP algorithm searches for a (possibly random) time $-T$ in the past such that, 
when the chain is simulated forward from every possible initial state at time $-T$ using the same random bits, all trajectories
\[
X_{-T}\to X_{-T+1}\to \cdots \to X_{-1}\to X_0
\]
coalesce into the same state at time~$0$.
Once such a coalescence time is identified, the final state $X_0\sim\pi$ is completely determined by the random bits used in the simulation, independent of the initial state at time~$-T$. Consequently, $X_0$ is a perfect sample from the stationary distribution $\pi$.


At the heart of CFTP lies the notion of a \emph{grand coupling},
which provides a unified probabilistic construction that simultaneously drives the evolution of the chain from all possible initial states.
Formally, for a Markov chain with transition kernel $P(x,y) = \Pr{X_{t+1}=y \mid X_t=x}$ on state space $\Omega$, a grand coupling is a mapping
$g : \Omega \times [0,1] \to \Omega$,
such that for $U \sim \uniform[0,1]$ and every $x,y \in \Omega$,
\begin{align}\label{eq:grand-coupling-intro}
  \Pr{g(x,U)=y}=P(x,y).   
\end{align}
In other words, $g(\cdot,U)$ simulates a single Markov chain transition under shared randomness $U$, thereby coupling all trajectories in a common probability space.
This notion extends naturally to the $T$-step transitions, represented by a grand coupling  
$F : \Omega \times [0,1] \to \Omega$
for the evolution $X_{-T} \to X_0$.
The CFTP algorithm identifies a time $-T$ in the past and terminates once $F(\cdot, U)$ becomes constant,
that is, when all trajectories coalesce into a single state at time~$0$.


Designing grand couplings that achieve fast coalescence is therefore the central problem in CFTP.
Beyond CFTP, grand couplings have found applications in a variety of contexts, including derandomization of MCMC~\cite{feng2023towards}, analytic stability (absence of complex zeros)~\cite{liu2025phase}, and local algorithms for sampling~\cite{liu2025local}.
Despite being a natural and fundamental concept, their systematic design and analysis remain poorly understood.
While fast-coalescing grand couplings immediately imply rapid mixing of the underlying Markov chain, 
essentially nothing is known about the converse in general.
This raises a fundamental question:
\begin{center}
    \emph{Does fast coalescence of grand couplings require stronger conditions than rapid mixing?}
\end{center}

\subsection{Bounding chains for \texorpdfstring{$q$}{q}-colorings}
\label{subsec:local-grand-coupling}

A key algorithmic challenge in implementing CFTP lies in detecting coalescence,
that is, determining whether the induced random mapping on the state space $\Omega$ has become a constant function.
Formally, let $\Omega_t\subseteq\Omega$ denote the set of all possible states of $X_t$ that can be reached from different initial configurations $X_{-T}\in\Omega$ under the same randomness.
Specifically, 
\[
\Omega_t \defeq \{\, g_t \circ g_{t-1} \circ \cdots \circ g_{-T}(\conf) : \conf \in \Omega \,\},
\]
where each $g_t=g_t(\cdot, U)$ denotes the grand coupling for the transition $X_t\to X_{t+1}$ with  $U\sim \uniform[0,1]$. 
Initially, $\Omega_{-T}=\Omega$, and clearly $\abs{\Omega_{t+1}}\le \abs{\Omega_t}$ since $g_t$ is a function. 
Coalescence occurs once $\abs{\Omega_t}=1$ for some $t\le 0$.
In certain special cases, most notably in monotone systems such as the ferromagnetic Ising model, where a natural partial order on $\Omega$ is preserved by the grand coupling, checking coalescence is straightforward. 
In contrast, for general systems, verifying coalescence is substantially more intricate and can be computationally intractable.

The \emph{bounding chains} method, introduced by Huber~\cite{huber1998exact,huber2004perfect}, provides an efficient way to detect coalescence by tracking a ``bounding box'' of the subset $\Omega_t$ in the space $\Omega\subseteq[q]^V$.
The bounding chain maintains, for each vertex $v \in V$, a \emph{bounding list} $\mathcal{L}_t(v) \subseteq [q]$ of colors, which contains all colors that vertex $v$ may take under any configuration in $\Omega_t$.  
The corresponding bounding box is given by
\[
\mathcal{B}_t \defeq \bigotimes_{v \in V} \mathcal{L}_t(v).
\]
Initially, we set $\mathcal{L}_{-T}(v) = [q]$ for every vertex $v$.  
The bounding list configuration $\mathcal{L}_t=(\mathcal{L}_t(v))_{v\in V}$ evolves as a Markov chain:
at each time step $t\ge -T$, the next bounding list $\mathcal{L}_{t+1}$ is constructed from the current list $\mathcal{L}_{t}$, under the conservative assumption that every configuration $\conf\in\mathcal{B}_t$ may appear at time~$t$.
%
By induction, we have $\Omega_t\subseteq\mathcal{B}_t$ for any $t\ge -T$.
Thus, coalescence at time $t$ occurs if the bounding box has collapsed to the size $\abs{\mathcal{B}_t}=1$, or equivalently, when $\abs{\mathcal{L}_t(v)}=1$ for all $v\in V$.

The fast coalescence of bounding chains relies on the following standard contraction property.

\begin{condition}[Contraction of bounding chains]
\label{cond:bounding-chain-overall}
Let $(\mathcal{L}_t)$ be a bounding chain. At any time $t$, it holds
\begin{align}
\label{eq:potential-function-decay}
\mathbb{E}\!\left[\,\sum_{v \in V} \abs{\mathcal{L}_{t+1}(v)} \,\middle|\, \mathcal{L}_{t}\right]
\leq \sum_{v \in V} \abs{\mathcal{L}_t(v)}.
\end{align}
\end{condition}

Intuitively, this condition requires the total size of the bounding lists to be monotonically non-increasing in expectation at each step, even in the worst case.
All existing CFTP algorithms based on bounding chains satisfy this condition, including~\cite{huber1998exact,bhandari2020improved,jain2021perfectly}.
The seminal work of Huber~\cite{huber1998exact} proposed a bounding chain algorithm for uniform sampling of $q$-coloring that is fast coalescing under the condition $q \ge (1+o(1))\Delta^2$.  
For many years, this remained the best achievable bound for any CFTP algorithm for $q$-coloring, until recent breakthroughs~\cite{bhandari2020improved, jain2021perfectly} improved the condition for fast-coalescing bounding chains first down to $q >3\Delta$ and later to $q \ge  (\frac{8}{3}+o(1))\Delta$.

\subsection{Our results}
We establish an asymptotically optimal threshold of $q = 2.5\Delta$ for CFTP algorithms based on bounding chains for proper $q$-colorings.  
In particular, we design an efficient CFTP algorithm that samples uniform proper $q$-colorings under the improved condition
 $q \ge (2.5 + o(1))\Delta$.



\begin{theorem}[Upper bound]
\label{thm:main}
There exists a CFTP algorithm that, given an undirected graph $G = (V, E)$ with maximum degree $\Delta \ge 3$ and
$q > 2.5\Delta + 2\sqrt{(\log \Delta + 1)\Delta}$,
outputs a uniformly random proper $q$-coloring of $G$ in expected time 
$\tilde{O}(n\Delta^2)$,
where $\tilde{O}(\cdot)$ hides poly-logarithmic factors.
\end{theorem}

The algorithm in \Cref{thm:main} is based on a new construction of bounding chains for the Glauber dynamics on uniform proper $q$-colorings, which coalesce rapidly using substantially fewer colors.

This result advances the state of the art for CFTP algorithms and, more broadly, for all perfect sampling algorithms for $q$-colorings that achieve exponential convergence.

\begin{remark}[Perfect sampling via deterministic counting or approximate sampling]
\label{remark:discuss-on-GLMP}
Perfect sampling can also be achieved via generic reductions to deterministic approximate counting or approximate sampling.  
In the classical work of Jerrum, Valiant, and Vazirani~\cite{jerrum1986random}, a general reduction from perfect sampling to deterministic approximate counting (FPTAS) was established using a rejection-sampling framework.  
More recently, G\"{o}bel, Liu, Manurangsi, and Pappik~\cite{Gobel2024Perfect} presented a reduction from perfect sampling to rapidly mixing Markov chains: 
once the chain has mixed sufficiently well (so that the total variation error is exponentially small), a JVV-style filtering step (implemented via exponentially exhaustive enumeration) is applied with exponentially small probability, keeping the overall expected runtime polynomial.

Through these reductions, the known thresholds for polynomial-time perfect sampling of $q$-colorings have been improved to $q > 2\Delta$ via deterministic approximate counting~\cite{liu2019deterministic} and to $q > 1.809\,\Delta$ via rapidly mixing Markov chains~\cite{carlson2025flip}.

Despite these advances, the coupling-from-the-past (CFTP) paradigm remains of great theoretical and practical significance. 
It provides a systematic framework for designing Las Vegas samplers with \emph{exponentially convergent} runtime.
Specifically, letting $T$ denote the random runtime, we have
\[
\Pr{T\ge k\cdot \mathbb{E}[T]}\le \exp(-\Omega(k)),\quad\text{for all }k\ge 1.
\]
where in our case the expected runtime satisfies $\mathbb{E}[T] = \tilde{O}(n\Delta^2)$.
This sharp concentration arises from the efficient \emph{certificate} of perfect sampling inherently provided by CFTP through coalescence, a key advantage over the certificates in~\cite{jerrum1986random,Gobel2024Perfect}, which rely on deterministic counting.

Beyond its role in perfect sampling, the grand coupling is a natural and fundamental construct for Markov chains, 
and has found implications in several key aspects, including studies of derandomization~\cite{feng2023towards}, phase transitions~\cite{liu2025phase}, and local algorithms for sampling~\cite{liu2025local}.
It is therefore of central theoretical interest to understand the algorithmic power of grand couplings.
\end{remark}



We show that the upper bound in \Cref{thm:main} is asymptotically tight by establishing a matching lower bound for bounding-chain-based grand couplings satisfying the contraction property in \Cref{cond:bounding-chain-overall}.

\begin{theorem}[Lower bound]\label{thm:no-bounding-chain}
Let $\Delta \ge 3$, $q < 2.5\Delta-1$, and let $n$ be sufficiently large.  
There exists a $\Delta$-regular graph $G = (V, E)$ with $|V| = n$, together with a bounding list configuration $\mathcal{L} = (\mathcal{L}(v))_{v \in V}$, such that for any bounding chain for the Glauber dynamics on proper $q$-colorings, the contraction property~\eqref{eq:potential-function-decay} in \Cref{cond:bounding-chain-overall} is violated when the bounding chain is in state $\mathcal{L}_t = \mathcal{L}$.
%
\end{theorem}

Together, our results establish a tight bound for CFTP for $q$-colorings based on bounding chains.
They show that contractive bounding chains require a significantly stronger condition on the number of colors $q$ than is needed for rapid mixing of the underlying Markov chain.
In particular, the critical threshold $q>2.5\Delta$ identified here is substantially stronger than the classical Dobrushin condition for $q$-colorings, which holds when $q>2\Delta$ and corresponds to the worst-case contraction condition in path coupling~\cite{jerrum1995very,salas1997absence,bubley1997path}, thereby guaranteeing rapid mixing.

\begin{remark}
In~\cite{jain2021perfectly}, Jain, Sah, and Sawhney also identified $q = 2.5\Delta$ as a natural barrier for their bounding chain construction, reflecting a limitation of the algorithmic framework adopted in~\cite{bhandari2020improved,jain2021perfectly}.
In contrast, our lower bound shows that this threshold is inherent to all contractive bounding chains: no grand coupling based on such constructions can surpass it.
Both our general lower bound and the optimal bounding chain design arise from a new characterization of bounding-chain–based grand couplings as an optimization problem, detailed in \Cref{sec:coupling}.
\end{remark}

\subsection{Overview of conceptual and technical contributions}

In the CFTP framework, the random evolution $X_{-T}\to X_{0}$ of the Markov chain $(X_t)$ is captured by a random mapping $F : \Omega \to \Omega$, generated by a grand coupling $F = F(\cdot, U)$ using a random seed $U\in[0,1]$.
Equivalently, $F$ can be expressed as a composition of single-step couplings, $F = g_{-1} \circ g_{-2} \circ \cdots \circ g_{-T}$,
where each $g_t = g_t(\cdot, U_t)$ corresponds to the grand coupling of the one-step transition at time $t$, determined by an independent random seed $U_t\sim\uniform[0,1]$.
Conceptually, the global seed $U$ can thus be viewed as an infinite sequence of independent random variables $(U_{-1}, U_{-2}, \ldots)$, each determining one step of the coupled evolution.

When the Markov chain $(X_t)$ is a single-site dynamics, each transition $g_t$ is specified by a pair $(v_t, f_{v_t})$,
where $v_t \in V$ denotes the vertex selected for update at time $t$,
and $f_{v_t}$ is a local random update function that recolors $v_t$ based on the current coloring of its neighbors.
%
Once the sequence of update vertices $v_{-1},v_{-2},\ldots$ is fixed,
the CFTP process is fully determined by these local updates.

Conceptually, the bounding-chain construction can then be viewed as a \emph{localization} of this global coupling:
rather than tracking the evolution of the entire configuration space $\Omega \subseteq [q]^V$,
it maintains for each vertex $v$ a local bounding list $\mathcal{L}(v) \subseteq [q]$ of possible colors,
which compactly encodes all configurations consistent with the current stage of the coupled evolution.

\begin{remark}
\label{remark:coupling-random-vt}
In the original formulation of bounding chains for single-site dynamics~\cite{huber1998exact},
the update vertex is chosen uniformly at random from all vertices.
Subsequent developments~\cite{bhandari2020improved, jain2021perfectly} introduced more flexible update schedules to accelerate coalescence.
For correctness, however, the update schedule must be specified \emph{a priori},
independent of the evolution of the bounding lists during execution.

In the bounding-chain framework, the sequence of update vertices $(v_{-1}, v_{-2}, \ldots)$
forms part of the shared randomness and is \emph{identically coupled} across all global configurations.
That is, every coupled trajectory uses the same pre-specified sequence of update locations.
This identically coupled choice of the update sequence $(v_t)$ ensures that each update $(v_t, f_{v_t})$ acts purely locally
and that the bounding-chain coupling remains well defined regardless of how the bounding lists evolve.
\end{remark}

\paragraph{\textbf{Formulation of grand coupling via bounding chains}}
Building on the local perspective above, we introduce a unified formulation of the bounding-chain method.
In this framework, a global grand coupling of the Markov chain is specified via a collection of \emph{local grand couplings} defined over the bounding boxes of neighborhood configurations, as formalized in~\Cref{def:grand-coupling-on-bounding-chain}.

This subsumes all prior constructions of bounding chains, including those based on the standard uniform random update schedule as well as more structured, designed schedules, and enables a formal analysis of the fundamental limitations of the bounding-chain approach.



\medskip
\paragraph{\textbf{Lower bound: limitations of bounding chains.}}
\Cref{thm:no-bounding-chain} establishes a lower bound for bounding chains for $q$-coloring Glauber dynamics satisfying the standard contraction condition (\Cref{cond:bounding-chain-overall}).
%
This applies to any schedule, as long as updates follow the Glauber dynamics.
It stems from a fundamental obstruction to the coalescence of bounding chains: the reduction of bounding lists from size two to one.

\begin{theorem}[Lower bound: 2-to-1 contraction]\label{thm:no-coupling-informal}
Let $\Delta \ge 3$, $q < 2.5\Delta-1$, and let $n$ be sufficiently large.  
There exist a $\Delta$-regular graph $G = (V, E)$ with $|V|= n$ and a bounding list configuration $\mathcal{L}=(\mathcal{L}(v))_{v\in V}$ with $|\mathcal{L}(v)| = 2$ for all $v \in V$,
such that for any grand coupling of $q$-coloring Glauber dynamics given~$\mathcal{L}$, 
the updated $\mathcal{L}'$ satisfies: there exists a $v \in V$ with 
$\mathbb{E}\left[|\mathcal{L}'(v)|\right] > 2$,
while for all $u \in V\setminus\{v\}$, $|\mathcal{L}'(u)|=2$.
\end{theorem}

The construction in \Cref{thm:no-coupling-informal} 
results in an expected increase in the total size
$\sum_{v}|\mathcal{L}'(v)|$ relative to $\sum_{v}|\mathcal{L}(v)|$, violating \Cref{cond:bounding-chain-overall}.
Thus, \Cref{thm:no-bounding-chain} follows from \Cref{thm:no-coupling-informal}.  
Moreover, since adding colors to $\bl$ can only increase size of $\bl'$, it blocks coalescence from any bounding list sizes~$>1$.

\medskip

\paragraph{\textbf{Upper bound: optimal design of grand coupling}}

We formalize the design of a grand coupling as an explicit optimization problem (\Cref{prob:optimal-grand-coupling}).
To make the problem tractable, we introduce a simplified variant (\Cref{prob:optimal-grand-coupling-2}) that captures the key design issues underlying prior bounding-chain algorithms.


Solutions to this simplified problem can be characterized via a linear program, which yields a grand coupling that optimally solves \Cref{prob:optimal-grand-coupling-2}.
%
This optimal construction, called \seeding{}, provides a key component: Combined with several existing bounding-chain building blocks, including the \compress{} coupling~\cite{bhandari2020improved} and the \disjoint{} coupling~\cite{jain2021perfectly}, we finally obtain a fast coalescing grand coupling for $q$-coloring when $q > (2.5+o(1))\Delta$, matching, up to lower-order terms, the lower bound established for any bounding-chain–based CFTP.

\subsection{Related work}
\label{sec:ralated-work}
The critical condition for rapid mixing of local Markov chains on proper $q$-colorings, such as Glauber dynamics, has been a central question in the study of sampling algorithms.
A classic result of Jerrum~\cite{jerrum1995very} (and independently, Salas and Sokal~\cite{salas1997absence}) established that when $q > 2\Delta$, the Glauber dynamics mixes in $O(n \log n)$ time.
This foundational work laid the groundwork for coupling-based analyses, notably inspiring the path coupling technique of Bubley and Dyer~\cite{bubley1997path} and the coupling criterion of Hayes~\cite{hayes2006simple}.
Vigoda~\cite{vigoda1999improved} introduced a coupling for the \emph{flip dynamics} and proved an $O(n^2)$ mixing-time bound for the Glauber dynamics when $q > \tfrac{11}{6}\Delta$.
This threshold remained the state of the art for nearly two decades until Chen et al.~\cite{chen2019improved} further refined it to $q > \left(\tfrac{11}{6} - \varepsilon_0\right)\Delta$ for an explicit $\varepsilon_0 = 1/84000$.
Most recently, Carlson and Vigoda~\cite{carlson2025flip} improved the bound to $q > 1.809\Delta$, which remains the best-known rapid-mixing threshold for general graphs.

Under additional structural assumptions, such as a lower bound on the girth, sharper thresholds are known.
For instance, on triangle-free graphs, the Glauber dynamics mixes optimally whenever $q > \alpha^* \Delta$ for some constant $\alpha^* \approx 1.763$; see~\cite{hayes2006coupling,chen2021rapid,feng2022rapid,CLV21,jain2022spectral}.
Even stronger results hold for graphs of larger girth~\cite{hayes2003non,dyer2013randomly,chen2023strong}.
In particular, the recent breakthrough of Chen et al.~\cite{chen2023strong} establishes an almost tight condition $q \ge \Delta + 3$ for an $O(n \log n)$ mixing time of the Glauber dynamics on graphs with constant maximum degree $\Delta = O(1)$ and sufficiently large girth.
%


The conjectured optimal threshold for rapid mixing is $q > \Delta$, corresponding to the uniqueness threshold of the Gibbs measure on the infinite $\Delta$-regular tree~\cite{jonasson2002uniqueness}.
In contrast, approximate sampling from $q$-colorings becomes computationally intractable when $q < \Delta$, unless $\mathrm{NP} = \mathrm{RP}$~\cite{galanis2015inapproximability}.

In a separate line of work focused on deterministic approximate counting, Gamarnik and Katz~\cite{gamarnik2012correlation} gave a polynomial-time approximate counting algorithm for proper $q$-colorings when $q > 2.84\Delta$ for bounded-degree graphs. 
Subsequent improvements were obtained by Lu and Yin~\cite{lu2013improved} ($q > 2.58\Delta$) and by Liu, Sinclair, and Srivastava~\cite{liu2019deterministic} ($q > 2\Delta$).

\medskip
\paragraph{\textbf{Related works on CFTP}}
Coupling from the past (CFTP) is a central technique for perfect sampling and has been widely applied across statistical physics and combinatorial models.
Classical applications include the Ising and Potts models~\cite{haggstrom1998ising,huber2004perfect}, while extensive developments on graph colorings have leveraged the bounding-chain framework~\cite{huber2004perfect,huber2016perfect,bhandari2020improved,jain2021perfectly}.
More recently, CFTP methods have been extended to the sampling Lov\'{a}sz Local Lemma, leading to perfect samplers for solutions to atomic constraint satisfaction problems (CSPs)~\cite{he2021perfectsamplingatomiclovasz,qiu2022perfectsamplerhypergraphindependent}.

\medskip
\paragraph{\textbf{Other perfect sampling paradigms}}
Beyond CFTP, several alternative paradigms for perfect sampling have been developed, including the \emph{random recycler}~\cite{fill2000randomness}, \emph{partial rejection sampling}~\cite{wilson1996generating,cohn2002sink,guo2019uniform,feng2021dynamic}, and approaches based on the \emph{Bayes filter}~\cite{feng2022perfect}.
More recently, \emph{local perfect samplers} have been introduced, which generate samples from marginal distributions at a local cost.
Examples include the \emph{lazy depth-first sampler} (also known as the Anand–Jerrum algorithm)~\cite{AJ22,HWY22a} and the \emph{coupling towards the past} framework~\cite{feng2023towards,liu2025phase,liu2025local}, the latter extending the philosophy of coupling from the past to more general local settings.

\section{Preliminaries} 
\label{sec:preliminaries}

\subsection{Notations for coloring}
Let $G=(V,E)$ be an undirected graph and $[q] = \{1,2,\dots,q\}$ a set of $q\ge 2$ colors.
The configuration space of proper $q$-colorings of $G$ is
\[
\Omega=\{\omega:V\to[q]\mid \omega(u)\neq\omega(v)\text{ for all }\{u,v\}\in E\}.
\]
For each vertex $v$, let $\Gamma(v) \defeq \{u \in V \mid \{u,v\} \in E\}$ denote $v$'s neighborhood.

\subsection{Glauber dynamics}

A \emph{single-site dynamics} on $\Omega$ is defined by a predetermined \emph{update schedule}, which is a (possibly random) sequence of vertices $(v_t)$, together with a collection of \emph{local update distributions}:
$$\left\{p_v^\conf:v\in V,\conf\in[q]^{\Gamma(v)}\right\},$$
where $p_v^\conf$ specifies a distribution over $[q]$ that depends only on the vertex $v$ and the configuration $\conf$ of its neighbors.
The resulting dynamics $(X_t)$ evolves as follows: at each time step $t$, the vertex $v_t$ is selected according to the update schedule, and its color is updated by sampling from $p_{v_t}^{X_t(\Gamma(v_t))}$, where $X_t(\Gamma(v_t))$ denotes the restriction of the current configuration $X_t$ to the neighborhood $\Gamma(v_t)$ of $v_t$.
r
The \emph{Glauber dynamics} for proper $q$-colorings is a single-site dynamics in which, for each vertex $v \in V$ and neighborhood coloring $\conf \in [q]^{\Gamma(v)}$, the local update distribution at $v$ is uniform over the set of colors available to $v$ given $\conf$:
\begin{align}\label{eq:glauber-dynamics-update}
    p_v^{\conf} = \uniform([q] \setminus \conf),
\end{align}
where $\uniform(S)$ denotes the uniform distribution over a set $S$.
Here, by slight abuse of notation, we treat the tuple $\conf$ equivalently as a set, so that $[q] \setminus \conf=[q] \setminus \{\conf(u) : u \in \Gamma(v)\}$.


\subsection{Bounding chains}
\label{sec:preliminary-bounding-chains}
For a single-site dynamics $(X_t)$ on a state space $\Omega \subseteq [q]^V$ with update schedule $(v_t)$ and local update distributions ${p_v^\conf : v \in V,, \conf \in [q]^{\Gamma(v)}}$,
let $g_t : \Omega \times [0,1] \to \Omega$ denote a grand coupling for the transition $X_t \to X_{t+1}$, as defined in~\eqref{eq:grand-coupling-intro}.

Suppose that the randomness in the update schedule $(v_t)$ is \emph{identically coupled} across all trajectories, i.e., at each time $t$, every coupled trajectory updates the same randomly chosen vertex $v_t$.

Under this natural assumption, the coupling $g_t(\cdot, U)$ becomes a local function, such that: 
\begin{itemize}
    \item the image $g_t(\conf, U)$ depends only on $\conf\in\Omega$ restricted to $\Gamma(v_t)$;
    \item the image $g_t(\conf, U)$ differs from $\conf\in\Omega$ only at the update vertex $v_t$.
\end{itemize}
Hence, without loss of generality, we may write 
$$g_t(\conf, U)(v_t)=g_t\left(\conf_{\Gamma(v_t)}, U\right)(v_t),$$ 
since the outcome of $g_t(\conf, U)$ at $v_t$ depends solely on the neighborhood configuration $\conf_{\Gamma(v_t)}$.

This locality property of grand couplings for single-site dynamics naturally motivates the definition of \emph{bounding chains}.
A bounding chain $(\bl_t)$ is naturally constructed from a grand coupling $(g_t)$ of the single-site dynamics.
At each time step $t$, the bounding chain maintains, for every vertex $v \in V$, a \emph{bounding list} (or \emph{bounding set}) $\bl_t(v) \subseteq [q]$ representing the possible colors that $v$ may take at time~$t$ under the coupled evolution, inferred based on the bounding lists.

Formally, given the shared random seed $U_t$ and the update vertex $v_t$, the bounding chain evolves as:
\begin{align}\label{eq:transition-bounding-list}
    \bl_{t+1}(v)
=
\begin{cases}
    \set{g_t(\conf,U_t)(v):\conf\in\bigotimes_{u\in\Gamma(v)}\bl_t(u)}, & \text{if }v=v_t,\\
    \bl_{t}(v) & \text{otherwise}.
\end{cases}
\end{align}
In words, when vertex $v_t$ is updated, the new bounding list $\bl_{t+1}(v_t)$ consists of all possible colors that $v_t$ can take, given the current bounding lists of its neighbors and the shared randomness $U_t$ at time~$t$.
For all other vertices, the bounding lists remain unchanged.
In practice, one may instead maintain a superset of the set defined in~\eqref{eq:transition-bounding-list}, as it can be easier to compute. 
Such relaxations preserve the correctness of sampling, though they may affect the efficiency of coalescence.




In coupling from the past (CFTP), the bounding chain is initialized with full uncertainty: $\bl_{-T}(v) = [q]$ for all $v \in V$.
As the chain $(\bl_t)$ evolves forward from time $-T$ to time~$0$, the bounding lists ${\bl_t(v)}_{v\in V}$ gradually shrink as uncertainty about each vertex’s color decreases.
When all bounding lists collapse to singletons, i.e., $|\bl_t(v)| = 1$ for every $v \in V$, 
the coupled single-site dynamics have coalesced to a single configuration.
At this point, the bounding-chain-based CFTP procedure outputs a perfect sample from the stationary distribution of the Glauber dynamics.

\section{Grand Coupling via Bounding Chains}
\label{sec:coupling}



In this section, we formalize a unified framework for grand couplings arising from bounding chains.
We begin by formally defining the notion of a grand coupling.




\begin{definition}[Grand coupling]
\label{def:grand-coupling}
Let $\Omega$ be a finite sample space, and let $\mathcal{P}$ be a family of probability distributions on $\Omega$.
A measurable function $f: \mathcal{P} \times [0,1] \to \Omega$ is called a \emph{grand coupling} on $\mathcal{P}$ if, 
for a uniform random variable $U \sim \uniform[0,1]$, it holds that for every $p \in \mathcal{P}$ and every $y \in \Omega$,
\[
\Pr{f(p,U) = y} = p(y).
\]
Equivalently, for each fixed $p \in \mathcal{P}$, the random variable $f(p,U)$ has the distribution $p$.
\end{definition}

A grand coupling thus provides a unified deterministic rule that simultaneously generates a family of probability distributions on the same sample space $\Omega$ using shared randomness $U \sim \uniform[0,1]$.
In particular, the grand coupling for a Markov chain transition $P$ defined in~\eqref{eq:grand-coupling-intro} is a special case of \Cref{def:grand-coupling},
where the family $\mathcal{P}$ consists of the update distributions $P(x,\cdot)$ for all $x \in \Omega$.

\subsection{Local grand coupling via bounding chains}
We introduce a general formulation for grand couplings constructed using bounding chains, which applies to single-site dynamics.  
Let $G=(V,E)$ be a graph and $[q]=\{1,2,\ldots,q\}$ a set of $q\ge 2$ colors.  
Recalle that a single-site dynamics on a state space $\Omega \subseteq [q]^V$ is specified via a collection of local update distributions:
$\left\{p_v^\conf:v\in V,\conf\in[q]^{\Gamma(v)}\right\}$.

In the bounding-chain framework, we maintain for each vertex $v \in V$ a \emph{bounding list} $\bl(v) \subseteq [q]$, representing the set of possible colors that $v$ may take across coupled trajectories.  
The grand coupling is then constructed based on these bounding lists.

\begin{definition}[Grand coupling via bounding chain]
\label{def:grand-coupling-on-bounding-chain}
Let $G=(V,E)$ be a graph and $q \ge 2$.  
Consider a single-site dynamics on the state space $\Omega \subseteq [q]^V$ with local update distributions 
$\left\{p_v^\conf:v\in V,\conf\in[q]^{\Gamma(v)}\right\}$.

Fix a vertex $v \in V$ and a bounding list configuration $(\bl(u))_{u \in \Gamma(v)}$, with each $\bl(u) \subseteq [q]$.  
    Define
    \begin{align}\label{eq:bounding-box-local-distribution}
              \mathcal{P}_v \defeq \left\{ p_v^\conf : \conf \in \bigotimes_{u\in\Gamma(v)}\bl(u)\right\},
    \end{align}
the family of all possible local update distributions at $v$ consistent with the bounding lists.
    
    A \emph{local grand coupling} at vertex $v$ given the bounding list $(\bl(u))_{u \in \Gamma(v)}$ is a measurable mapping
    $$f_v : \mathcal{P}_v \times [0,1] \to [q]$$
such that $f_v$ is a valid grand coupling on $\mathcal{P}_v$ in the sense of \Cref{def:grand-coupling}.
\end{definition}

Previous bounding-chain based algorithms
(e.g.,~\cite{bhandari2020improved,jain2021perfectly})
implicitly employ two stages.
These stages can be unified in our framework for grand couplings via bounding chains.



\begin{itemize}
\item \emph{Generating stage:}  
Once a random seed $U \sim \uniform[0,1]$ is generated, the function $f_v$ determines, for every $p \in \mathcal{P}_v$, a possible update outcome of $v$.  
The collection of all outcomes $\{ f_v(p, U) : p \in \mathcal{P}_v \}$ forms the new bounding list $\bl'(v)$ of vertex $v$, representing all colors that could be assigned under the current randomness.

\item \emph{Decoding stage:}  
Once the actual neighbor configuration $\conf \in [q]^{\Gamma(v)}$ (and thus the corresponding update distribution $p_v^\conf \in \mathcal{P}_v$) is revealed, the updated color of $v$ is determined by $f_v(p_v^\conf, U)$.
\end{itemize}

This formulation unifies the generating–decoding paradigm into a single mathematical framework, enabling an analysis of the optimality of grand coupling designs.

\subsubsection{Grand coupling for Glauber dynamics}

Now consider the Glauber dynamics on proper $q$-colorings.  
For each vertex $v$ and neighborhood configuration $\conf \in [q]^{\Gamma(v)}$, the local update distribution $p_v^{\conf}$ at $v$ given $\conf$ is uniform over the set of colors available to $v$ given $\conf$, that is, $p_v^{\conf} = \uniform([q] \setminus \conf)$, where we adopt the abuse of notation as in~\eqref{eq:glauber-dynamics-update}.

Under this specialization, the grand coupling~$f_v$ in \Cref{def:grand-coupling-on-bounding-chain} admits an \emph{interval-based} model:
The preimages of $f_v(p, \cdot)$ induce a measurable partition of $[0,1]$ corresponding to the outcomes of~$p$.  
%
Let $k \defeq |[q]\setminus\conf|$ denote the number of available colors at~$v$.
The mapping $f_v(p_v^{\conf}, \cdot)$ defines a partition of the unit interval $[0,1]$ into $k$ subsets of equal measure~$1/k$, each labeled by a distinct available color.  
A shared random seed $U \sim \uniform[0,1]$ specifies a point in~$[0,1]$,  
and the update color $f_v(p_v^{\conf}, U)$ is the color assigned to that point in the partition corresponding to~$p_v^{\conf}$.

%


\begin{example}
    Suppose $q = 5$, and vertex~$v$ has two neighbors with bounding lists $\{1,2\}$ and $\{2,3\}$, respectively. 
Then the possible sets of colors used by these two neighbors are $\set{2,3}, \set{1,3}, \set{1,2}, \set{2}$;
and the corresponding sets of available colors at~$v$ are their complements $\{1,4,5\},\{2,4,5\},\{3,4,5\},\{1,3,4,5\}$.
In the interval-based model, the grand coupling~$f_v$ associates each case with a colored partition of~$[0,1]$, where the available colors are distributed uniformly over the interval. 
Given the shared randomness $U \in [0,1]$, a new color for~$v$ is determined in each case by the color at position~$U$. 
The updated bounding list~$\bl'(v)$ at vertex~$v$ is then the set of all such colors across all cases.
This is illustrated in \Cref{fig:grand-coupling-intervals}.
\end{example}
%

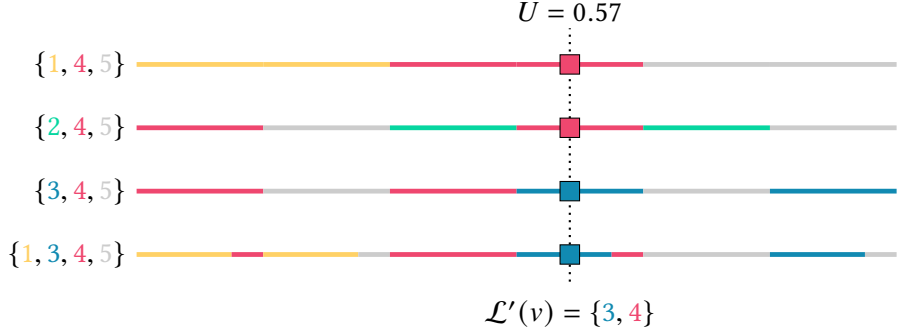
\begin{figure}[ht]
\centering
\begin{tikzpicture}[x=10cm, y=1.2cm]

\definecolor{color1}{HTML}{FFD166} 
\definecolor{color2}{HTML}{06D6A0} 
\definecolor{color3}{HTML}{118AB2}
\definecolor{color4}{HTML}{EF476F}
\definecolor{color5}{HTML}{CCCCCC}

\node[below] at (0.57, 3.2) {$U=0.57$};
\draw[dotted, thick] (0.57,0) -- (0.57,2.8);

\node[left] at (0,2.4) {$\{\textcolor{color1}{1}, \textcolor{color4}{4}, \textcolor{color5}{5}\}$};
\draw[line width=1.8, color1] (0.00,2.4) -- (1/6,2.4);
\draw[line width=1.8, color1] (1/6,2.4) -- (2/6,2.4);
\draw[line width=1.8, color4] (2/6,2.4) -- (3/6,2.4);
\draw[line width=1.8, color4] (3/6,2.4) -- (4/6,2.4);
\draw[line width=1.8, color5] (4/6,2.4) -- (5/6,2.4);
\draw[line width=1.8, color5] (5/6,2.4) -- (1.0,2.4);
\node[draw=black, fill=color4] at (0.57,2.4) {};

\node[left] at (0,1.7) {$\{\textcolor{color2}{2}, \textcolor{color4}{4}, \textcolor{color5}{5}\}$};
\draw[line width=1.8, color4] (0.00,1.7) -- (1/6,1.7);
\draw[line width=1.8, color5] (1/6,1.7) -- (2/6,1.7);
\draw[line width=1.8, color2] (2/6,1.7) -- (3/6,1.7);
\draw[line width=1.8, color4] (3/6,1.7) -- (4/6,1.7);
\draw[line width=1.8, color2] (4/6,1.7) -- (5/6,1.7);
\draw[line width=1.8, color5] (5/6,1.7) -- (1.0,1.7);
\node[draw=black, fill=color4] at (0.57,1.7) {};

\node[left] at (0,1.0) {$\{\textcolor{color3}{3}, \textcolor{color4}{4}, \textcolor{color5}{5}\}$};
\draw[line width=1.8, color4] (0.00,1.0) -- (1/6,1.0);
\draw[line width=1.8, color5] (1/6,1.0) -- (2/6,1.0);
\draw[line width=1.8, color4] (2/6,1.0) -- (3/6,1.0);
\draw[line width=1.8, color3] (3/6,1.0) -- (4/6,1.0);
\draw[line width=1.8, color5] (4/6,1.0) -- (5/6,1.0);
\draw[line width=1.8, color3] (5/6,1.0) -- (1.0,1.0);
\node[draw=black, fill=color3] at (0.57,1.0) {};

\node[left] at (0,0.3) {$\{\textcolor{color1}{1},\textcolor{color3}{3}, \textcolor{color4}{4}, \textcolor{color5}{5}\}$};
\draw[line width=1.8, color1] (0.00,0.3) -- (1/6-1/24,0.3);
\draw[line width=1.8, color4] (1/6-1/24,0.3) -- (1/6,0.3);
\draw[line width=1.8, color1] (1/6,0.3) -- (2/6-1/24,0.3);
\draw[line width=1.8, color5] (2/6-1/24,0.3) -- (2/6,0.3);
\draw[line width=1.8, color4] (2/6,0.3) -- (3/6,0.3);
\draw[line width=1.8, color3] (3/6,0.3) -- (4/6-1/24,0.3);
\draw[line width=1.8, color4] (4/6-1/24,0.3) -- (4/6,0.3);
\draw[line width=1.8, color5] (4/6,0.3) -- (5/6,0.3);
\draw[line width=1.8, color3] (5/6,0.3) -- (1.0-1/24,0.3);
\draw[line width=1.8, color5] (1-1/24,0.3) -- (1,0.3);
\node[draw==black, fill=color3] at (0.57,0.3) {};

\node[below] at (0.57,-0.1) {$\mathcal{L}'(v)=\{\textcolor{color3}{3},\textcolor{color4}{4}\}$};

\end{tikzpicture}
\caption{
Each line illustrates how the unit interval~$[0,1]$ is partitioned and assigned to the available colors, 
thereby defining a uniform distribution over them using the shared source of randomness $U \sim \mathrm{Uniform}[0,1]$. 
A single random position~$U \in [0,1]$ determines the color selected for each case, and the union of these selected colors forms the updated bounding list~$\mathcal{L}'(v)$. 
%
In this example, for any choice of~$U$, the updated bounding list satisfies $|\mathcal{L}'(v)| = 2$.
}
\label{fig:grand-coupling-intervals}
\end{figure}

\subsubsection{Classical construction of grand coupling for \texorpdfstring{$q$}{q}-colorings}
\label{sec:warm-up}
We next illustrate how the formulation in \Cref{def:grand-coupling-on-bounding-chain} specializes to the classical grand coupling of~\cite{huber1998exact}.  
%
In this coupling, the shared randomness is a uniformly random permutation $\pi$ of $[q]$.  
For each vertex $v$ with neighborhood coloring $\conf\in[q]^{\Gamma(v)}$, the local rule $f_v$ outputs the first color in $\pi$ that does not appear in $\conf$.  
It is straightforward to verify that, for any $\conf$, this rule samples uniformly from the set of available colors $[q] \setminus \conf$.


In the interval-based interpretation, the unit interval $[0,1]$ is partitioned uniformly into $q!$ subintervals, each corresponding to a distinct permutation of $[q]$.  
For each neighborhood coloring $\conf$, and for each subinterval, let $\pi$ denote the permutation associated with that subinterval;  
the subinterval is then assigned the first color in $\pi$ that does not appear in $\conf$.
Given a shared random seed $U \sim \uniform[0,1]$, the subinterval containing $U$ determines the color selected for each local update distribution $p \in \mathcal{P}_v$.  
Ideally, the new bounding list at vertex $v$ is
\[
\bl'(v)=\{ f_v(p, U) : p \in \mathcal{P}_v \},
\]
where $\mathcal{P}_v$ contains all local update distributions given the current bounding lists, as defined in~\eqref{eq:bounding-box-local-distribution}.

Although $|\mathcal{P}_v|$ may be exponentially large, it is not necessary to evaluate $f_v(p,U)$ for every $p \in \mathcal{P}_v$.  
A key observation here is that including extra colors into $\bl'(v)$ does not affect correctness of CFTP (though it may slow down coalescence).
Let $\pi$ denote the permutation corresponding to the interval containing $U$.  
Since the number of distinct colors in a neighborhood satisfies $|\{\conf(u): u \in \Gamma(v)\}| \le \Delta$, at least one of the first $\Delta+1$ entries of $\pi$ is available for $v$.  
Hence, it suffices to take
\[
\{ \pi_1, \pi_2, \ldots, \pi_{\Delta+1} \},
\]
as the updated bounding list $\bl'(v)$, since $\{ f_v(p, U) : p \in \mathcal{P}_v \}\subseteq\{ \pi_1, \pi_2, \ldots, \pi_{\Delta+1} \}$.
%
%

For further refinement, observe that if a color $c$ does not appear in the bounding list $\bl(u)$ of any neighbor $u \in \Gamma(v)$, then $c$ cannot appear in any neighborhood coloring $\conf \in \bigotimes_{u \in \Gamma(v)} \bl(u)$.  
Let $C$ denote the set of all such colors.  
If there exists an index $i \le \Delta+1$ such that $\pi_i \in C$, it suffices to take
\[
\{ \pi_1, \pi_2, \ldots, \pi_i \}
\]
as the updated bounding list.  
Indeed, $\pi_i \in C$ guarantees that for every neighborhood coloring $\conf \in \bigotimes_{u \in \Gamma(v)} \bl(u)$, the grand coupling outcome $f_v(p_v^{\conf}, U)$ never appears after position $i$ in $\pi$.  
Hence, $\{ f_v(p, U) : p \in \mathcal{P}_v \}\subseteq \{ \pi_1, \pi_2, \ldots, \pi_i \}$.  

The grand coupling described above, together with this refinement, precisely recovers the bounding chain of~\cite{huber1998exact} for $q$-colorings.
The effectiveness of a bounding chain depends on how well it tracks the bounding list
$
\{ f_v(p, U) : p \in \mathcal{P}_v \} 
$.
Smaller bounding lists may result in a tighter bounding chain and thus faster coalescence in the CFTP process.
Subsequent works~\cite{bhandari2020improved, jain2021perfectly} refined this idea by designing more sophisticated grand couplings.  
Their algorithms partition the process into multiple \emph{phases}, each employing a distinct grand coupling, allowing the bounding lists to shrink gradually across phases.  
Despite differences in implementation, all such couplings admit a natural interpretation within the framework in Definition~\ref{def:grand-coupling-on-bounding-chain}.

Within this formulation, designing an effective grand coupling can be quantitatively phrased as minimizing the size of
$
\{ f_v(p, U) : p \in \mathcal{P}_v \}
$
for $U \sim \uniform[0,1]$.  
In the following sections, we analyze this objective under specific structural assumptions on the neighboring bounding lists $(\bl(u))_{u \in \Gamma(v)}$ and establish both upper and lower bounds for grand couplings within this model.

\subsubsection{An optimization problem on grand coupling}
As discussed above, the performance of a grand coupling is naturally quantified by the size of its updated bounding list:
\[
\bl'(v) = \{\, f_v(p, U) : p \in \mathcal{P}_v \,\}, \qquad U \sim \uniform[0,1].
\]
A smaller expected size of $\bl'(v)$ leads to faster coalescence in the coupling-from-the-past procedure.  
This motivates designing grand couplings that minimize the updated bounding list given the current neighborhood bounding lists.
We formalize this task as the following optimization problem.

\begin{problem}[Local grand coupling optimization]
\label{prob:optimal-grand-coupling}
Fix a vertex $v \in V$ and a bounding list configuration $(\bl(u))_{u \in \Gamma(v)}$, with each $\bl(u) \subseteq [q]$.  
%
Let 
\begin{align}\label{eq:bounding-glauber-dynamics-updates}
    \mathcal{P}_v = \left\{\uniform([q] \setminus \conf): \conf \in \bigotimes_{u\in\Gamma(v)}\bl(u)\right\}
\end{align}
denote the set of all possible local update distributions 
at $v$ in the Glauber dynamics on $q$-colorings. 

The objective is to design a local grand coupling $f_v : \mathcal{P}_v \times [0,1] \to [q]$ at vertex $v$ given the bounding lists $(\bl(u))_{u \in \Gamma(v)}$ (cf.\ \Cref{def:grand-coupling-on-bounding-chain}) that minimizes 
\[
\mathbb{E}_{U \sim \uniform[0,1]}\!\bigl[\,|\{ f_v(p, U) : p \in \mathcal{P}_v \}|\,\bigr].
\]
That is, $f_v$ seeks to minimize the expected size of the updated bounding list at $v$.
\end{problem}

\subsection{Lower bounds for grand coupling}
\label{sec:grand-coupling-lower-bound}
In this subsection we prove Theorem~\ref{thm:no-bounding-chain} by establishing a fundamental lower bound for Problem~\ref{prob:optimal-grand-coupling}.  
Our analysis reveals an intrinsic obstacle that limits the effectiveness of all known bounding-chain constructions~\cite{huber1998exact,bhandari2020improved,jain2021perfectly}: 
\begin{center}
    \emph{The difficulty of reducing every bounding list from size~$2$ to~$1$.}
\end{center}
This step represents the critical bottleneck in coalescence of bounding chains.  
To understand this limitation, we focus on the configuration where each neighbor $u \in \Gamma(v)$  
has a bounding list $\bl(u)$ of size~$2$.  
This setting captures the worst-case local uncertainty at vertex~$v$,  
and hence determines the strongest contraction that any grand coupling can possibly attain.  
Our lower-bound argument is therefore developed under this extremal configuration.



\begin{theorem} \label{thm:no-coupling-2}
Let $q < 2.5\Delta-1$, and let $v$ be a vertex with degree $|\Gamma(v)| = \Delta$.  
There exists a bounding list configuration 
$(\bl(u))_{u \in \Gamma(v)}$ with $|\bl(u)| = 2$ for all $u \in \Gamma(v)$
such that, for every local grand coupling $f_v : \mathcal{P}_v \times [0,1] \to [q]$ of the Glauber dynamics update at  $v$ given $(\bl(u))_{u \in \Gamma(v)}$, where $\mathcal{P}_v$ is as in~\eqref{eq:bounding-glauber-dynamics-updates}, 
\[
\mathbb{E}_{U \sim \uniform[0,1]}\!\bigl[\,|\{ f_v(p, U) : p \in \mathcal{P}_v \}|\,\bigr] > 2.
\]
\end{theorem}

\Cref{thm:no-coupling-2} is established by explicitly constructing a bounding-list configuration exhibiting a structural bottleneck, formalized as follows.

\begin{definition}[Worst-case triangle configuration]
\label{def:worst-case}
We say that two vertices $u, v \in V$ form a \emph{triangle} on a bounding list configuration $\bl$ 
if $|\bl(u)| = |\bl(v)| = 2$ and $|\bl(u) \cup \bl(v)| = 3$.  

For a vertex $v \in V$, we say that $\bl$ is in a \emph{worst-case configuration at~$v$}  
if the neighbors of~$v$ can be partitioned into disjoint pairs, each forming a triangle on~$\bl$.
\end{definition}

\begin{proof}[Proof of \Cref{thm:no-coupling-2}]

We consider the \emph{worst-case configuration at $v$} from Definition~\ref{def:worst-case}, which realizes the scenario where no grand coupling can reduce the expected bounding list size at $v$ to $2$ or below.

Without loss of generality, assume $\Delta$ is even, write $\Delta = 2m$ and express $q = 3m + r$ for some integer $r \ge 0$. 
Partition $v$'s neighbors into $m$ disjoint pairs. Index the pairs by $i=1,\dots,m$.  For the $i$-th pair in each side we assign bounding lists as follows:
the two vertices in the $i$-th pair receive bounding lists $\{3i-2,\,3i-1\}$ and $\{3i-1,\,3i\}$ respectively.  
The colors $\{1,2,...,3m\}$ are called the \emph{constrained} colors, while the remaining $r$ colors are called \emph{free} colors.

We first note that the construction above indeed satisfies Definition~\ref{def:worst-case}: each neighbor pair indexed by $i=1,\dots,m$ has bounding lists
$\{3i-2,3i-1\}$ and $\{3i-1,3i\}$, so the two vertices in the pair form a triangle in the sense of Definition~\ref{def:worst-case}.
Let $f_v$ be an arbitrary grand coupling for the Glauber update at $v$, and write
\[
\mathcal{L}_U(v) \;=\; \{\, f_v(p,U) : p \in \mathcal{P}_v \,\}
\]
for the bounding list produced by $f_v$ when the shared seed equals $U\sim\uniform[0,1]$.    

For $i=1,\dots,m$, define the following events, where all probabilities are over $U\sim\uniform[0,1]$:
\begin{align*}
\mathcal{A}_i &:= \{\, \mathcal{L}_U(v) \cap \{3i-2,3i\} \neq \varnothing \,\},\\
\mathcal{B}_i &:= \{\, 3i-1 \in \mathcal{L}_U(v) \,\},\\
\mathcal{C}_i &:= \{\, \{3i-2,3i\} \subseteq \mathcal{L}_U(v) \,\},
\end{align*}
and let $\mathcal{D} := \{\, \mathcal{L}_U(v) \cap \{3m+1,\dots,3m+r\} \neq \varnothing \,\}$
denote the event that at least one \emph{free} color (a color in $\{3m+1,\dots,3m+r\}$) appears in $\mathcal{L}_U(v)$.
We claim the following two lower bounds:
\begin{align}
    \label{prop:A} \Pr{\mathcal{A}_i} &\ge \dfrac{2}{m+r+1}, &\forall i\in\{1,2,...,m\};\\
    \label{prop:B} \Pr{\mathcal{B}_i} &\ge \dfrac{1}{m+r},   &\forall i\in\{1,2,...,m\};
\end{align}
and a covering property:
\begin{align}
    \label{prop:cover} 
    \Pr{\mathcal{C}_1 \,\vee\, \cdots \,\vee\, \mathcal{C}_m \,\vee\, \mathcal{D}} &= 1, \quad\text{ and consequently, }\sum_{i=1}^m \Pr{\mathcal{C}_i} \;+\; \Pr{\mathcal{D}} \;\ge\; 1. 
\end{align}
Assuming these inequalities, we obtain a lower bound on the expected bounding list size. 
Observe
\[
\Pr{3i-2 \in \mathcal{L}_U(v)} + \Pr{3i \in \mathcal{L}_U(v)}
  \;=\; \Pr{\mathcal{A}_i} + \Pr{\mathcal{C}_i}, \qquad \forall i\in\{1,2,...,m\};
\]
since the left-hand side counts the indicators of the two events and the right-hand side decomposes that sum into the indicator of ``at least one" plus the indicator of ``both". Hence
\[
\begin{aligned}
\mathbb{E}\bigl[\,|\mathcal{L}_U(v)|\,\bigr]
&= \sum_{c\in[q]} \Pr{c \in \mathcal{L}_U(v)} \\
&\ge \left(\sum_{c\in\{1,...,3m\}} \Pr{c \in \mathcal{L}_U(v)}\right) + \Pr{\mathcal{L}_U(v) \cap \{3m+1,\dots,3m+r\} \neq \varnothing}\\
&= \sum_{i=1}^m \bigl(\Pr{3i-2 \in \mathcal{L}_U(v)} + \Pr{3i-1 \in \mathcal{L}_U(v)} + \Pr{3i \in \mathcal{L}_U(v)}\bigr) + \Pr{\mathcal{D}}\\
&= \left(\sum_{i=1}^m \Pr{\mathcal{A}_i} \right) + \left(\sum_{i=1}^m \Pr{\mathcal{B}_i} \right) + \left(\sum_{i=1}^m \Pr{\mathcal{C}_i} \right) + \Pr{\mathcal{D}}.
\end{aligned}
\]
Applying the bounds from \eqref{prop:A}--\eqref{prop:cover} yields
\[
\mathbb{E}\bigl[\,|\mathcal{L}_U(v)|\,\bigr]
 \;\ge\; \frac{2m}{m+r+1} + \frac{m}{m+r} + 1.
\]
Since $q < 2.5\Delta$ implies $r < 2m$, it follows that
\[
\frac{2m}{m+r+1} + \frac{m}{m+r} + 1 \;>\; \frac{2m}{3m} + \frac{m}{3m} + 1 \;=\; 2.
\]
Hence for every $q < 2.5\Delta$ we have
$\mathbb{E}\bigl[\,|\mathcal{L}_U(v)|\,\bigr] \;>\; 2$
as claimed.

It remains to prove \eqref{prop:A}--\eqref{prop:cover}.  
These are simple consequences of the marginal correctness of the grand coupling and the explicit worst-case assignment of neighbor lists:

\begin{itemize}
\item \textbf{Proof of \eqref{prop:A}.} 
Fix $i\in\{1,\dots,m\}$.  For every pair $j\neq i$, assigns the two neighbors the colors $3j-2$ and $3j$, and for the $i$-th pair assigns both neighbors the color $3i-1$. Let $\conf$ be the neighborhood coloring constructed above.
Then $|\conf| = \Delta-1 = 2m-1$,
so the update distribution is $\uniform([q]\setminus \conf)$, which places mass
$1/(q-(2m-1))=1/(m+r+1)$ on each of the two colors $3i-2$ and $3i$. 

By marginal correctness for this fixed distribution, each of the events ``$3i-2\in\mathcal{L}_U(v)$'' and ``$3i\in\mathcal{L}_U(v)$'' has measure at least $1/(m+r+1)$, hence
\[
\Pr{\mathcal{A}_i} \;\ge\; \Pr{3i-2 \in \mathcal{L}_U(v)} + \Pr{3i \in \mathcal{L}_U(v)}
\;\ge\; \frac{2}{m+r+1}.
\]

\item \textbf{Proof of \eqref{prop:B}.}
Fix $i\in\{1,\dots,m\}$.  Let $\conf$ be the neighbor configuration that, for every pair $j$, assigns the two neighbors the colors $3j-2$ and $3j$.  Then $|\conf|=\Delta$, so the update distribution is $\uniform([q]\setminus \conf)$, which places mass
$1/(q-\Delta)=1/(m+r)$ on each middle color $3i-1$.  

By marginal correctness for this fixed distribution we obtain
\[
  \Pr{\mathcal{B}_i} \;\ge\; \Pr{3i-1 \in \mathcal{L}_U(v)} \;\ge\; \frac{1}{m+r}.
\]

\item \textbf{Proof of \eqref{prop:cover}.}
Suppose for contradiction that for some $U$ none of $\mathcal{C}_1,\dots,\mathcal{C}_m,\mathcal{D}$ occurs.  Then for every $i$ at most one of the colors $3i-2,3i$ belongs to $\mathcal{L}_U(v)$, and no free color appears in $\mathcal{L}_U(v)$.  For each pair $i$ choose neighbor colors from the lists $\{3i-2,3i-1\}$ and $\{3i-1,3i\}$ so as to ensure every element of $\mathcal{L}_U(v)\cap\{3i-2,3i-1,3i\}$ is included in the resulting neighborhood set; this is always possible when at most one of $\{3i-2,3i\}$ is present.  Doing this for all $i$ yields a neighborhood coloring $\conf$ with
$\mathcal{L}_U(v) \subseteq \omega.$
But the coupling output $f_v([q]\setminus \omega,U)$ must lie in $[q]\setminus \omega$, contradicting the definition of $\mathcal{L}_U(v)$ as the set of possible outputs at seed $U$.  Hence the assumption is false, so for every $U$ at least one of $\mathcal{C}_1,\dots,\mathcal{C}_m,\mathcal{D}$ occurs, thus
\[
\Pr{\mathcal{C}_1 \,\vee\, \cdots \,\vee\, \mathcal{C}_m \,\vee\, \mathcal{D}} \;=\; 1.
\]
\end{itemize}

Combining the three items completes the proof of Theorem~\ref{thm:no-coupling-2}.
\end{proof}

We now present a global version of the above lower bound, which extends the local obstruction in Theorem~\ref{thm:no-coupling-2} to the entire graph, thereby proving Theorem~\ref{thm:no-bounding-chain}.

\begin{theorem}[Worst-case lower bound]
\label{thm:global-no-coupling}
Let $\Delta \ge 3$, $q < 2.5\Delta-1$, and let $n$ be sufficiently large.  
There exists a $\Delta$-regular graph $G = (V, E)$ with $|V| = n$, together with a bounding list configuration $\bl=(\bl(u))_{u \in V}$ satisfying $|\bl(u)| = 2$ for all $u \in V$, such that
for every vertex $v \in V$ and every local grand coupling
$f_v : \mathcal{P}_v \times [0,1] \to [q]$  of the Glauber dynamics update at~$v$ given $\bl$, where $\mathcal{P}_v$ is as in~\eqref{eq:bounding-glauber-dynamics-updates}, 
\[
\mathbb{E}_{U \sim \uniform[0,1]}\!\bigl[\,|\{ f_v(p, U) : p \in \mathcal{P}_v \}|\,\bigr] > 2.
\]
Moreover, the worst-case configuration at each vertex (defined in \Cref{def:worst-case}) can be simultaneously realized across all vertices of~$G$.
\end{theorem}

\begin{proof}
Without loss of generality assume $\Delta$ is even and write $\Delta = 2m$.
We first construct a $\Delta$-regular graph on $2\Delta$ vertices that realizes the worst-case
bounding-list configuration at every vertex, and then obtain larger graphs by taking disjoint copies.

Let $G = (V_1 \cup V_2, E)$ be the complete bipartite graph with bipartition sets
$V_1$ and $V_2$, where $|V_1| = |V_2| = \Delta$.  Then $G$ is $\Delta$-regular.
Partition the vertices of $V_1$ into $m$ disjoint pairs and likewise partition $V_2$ into $m$ disjoint pairs.
Index the pairs by $i=1,\dots,m$.  For the $i$-th pair in each side we assign bounding lists as follows:
the two vertices in the $i$-th pair receive bounding lists $\{3i-2,\,3i-1\}$ and $\{3i-1,\,3i\}$ respectively.
Do this assignment independently for $V_1$ and $V_2$ (using the same indexing convention).

By construction each bounding list has size~$2$.  Moreover, for any fixed vertex $v\in V_1$ its neighborhood $\Gamma(v)=V_2$ is partitioned into the $m$ disjoint pairs on $V_2$, and each such pair forms a triangle in the sense of Definition~\ref{def:worst-case}: the two lists in the pair each have size~$2$ and their union has size~$3$.
Hence the collection of bounding lists on $G$ is in the worst-case configuration at every $v\in V_1$.
The same argument applies to every $v\in V_2$ because the construction on $V_1$ is symmetric.
Therefore every vertex of $G$ is in the worst-case configuration of Definition~\ref{def:worst-case}.

By Theorem~\ref{thm:no-coupling-2}, for any vertex $v$ in such a worst-case configuration
and for any grand coupling $f_v$ for the single-site Glauber update at $v$ we have
\[
\mathbb{E}_{U \sim \uniform[0,1]}\!\bigl[\,|\{ f_v(p, U) : p \in \mathcal{P}_v \}|\,\bigr] > 2.
\]
Thus the $2\Delta$-vertex graph $G$ constructed above satisfies the claimed property for every vertex.

Finally, for any sufficiently large $n$ that is a multiple of $2\Delta$ we obtain an $n$-vertex
$\Delta$-regular graph with the same property by taking the disjoint union of $n/(2\Delta)$
independent copies of $G$, and by carrying the same bounding-list assignment on each copy.
Every vertex in this disjoint union is in the same local worst-case configuration, hence the
conclusion of the theorem holds simultaneously for every vertex of the $n$-vertex graph.
This completes the proof.
\end{proof}

As noted in \Cref{sec:preliminary-bounding-chains}, any bounding-chain–based grand coupling selects a vertex $v$ at random in each step and then applies a local grand coupling at $v$.  
Thus \Cref{thm:global-no-coupling} is in fact a formal restatement of \Cref{thm:no-coupling-informal}, 
and the latter can be derived as a straightforward corollary of the former.

\section{Design of Grand Couplings}
\label{sec:coupling-design}

In this section, we explore the design of grand couplings with desirable contraction properties, a key ingredient for CFTP algorithms that succeed near the critical threshold.

\subsection{A simplified formulation for grand coupling design}
\label{sec:grand-coupling-seeding}
Based on the bounding-chain framework, we formulated in \Cref{prob:optimal-grand-coupling} an optimization problem for designing grand couplings.  

To obtain a more tractable and intuitive formulation, we approximate $\mathcal{P}_v$ by considering only the set of colors 
$\slv = \bigcup_{u \in \Gamma(v)} \bl(u)$ appearing in the neighborhood bounding list of $v$, which coarsely tracks the local uncertainty around $v$.  
This motivates the following simplified grand coupling design problem, which serves as an approximation of \Cref{prob:optimal-grand-coupling}.

\begin{problem}[Grand coupling design problem]
\label{prob:optimal-grand-coupling-2}
Fix a subset of colors $\slv \subseteq [q]$ and an integer $\Delta$.  
Let $\mathcal{S}_\Delta \defeq \{\, C \subseteq \slv : |C| \le \Delta \,\}$ denotes the collection of all subsets of $\slv$ of size at most~$\Delta$,
and define
\[
\mathcal{P}=\mathcal{P}(\slv_{\Delta}) \defeq \set{\, \uniform([q]\setminus C) : C \in \mathcal{S}_\Delta \,}.
\]

The objective is to design a grand coupling 
 $f : \mathcal{P} \times [0,1] \to [q]$
(cf.\ Definition~\ref{def:grand-coupling}) that minimizes
\[
\mathbb{E}_{U \sim \uniform[0,1]}\!\bigl[\,|\{\, f(p, U) : p \in \mathcal{P} \,\}|\,\bigr].
\]
\end{problem}

It is straightforward to verify that 
any grand coupling $f : \mathcal{P} \times [0,1] \to [q]$ that is a feasible solution to \Cref{prob:optimal-grand-coupling-2} on an instance $(\slv, \Delta)$
is also a feasible solution to \Cref{prob:optimal-grand-coupling} for any vertex $v \in V$ with neighborhood bounding list configuration $(\bl(u))_{u \in \Gamma(v)}$ satisfy $|\Gamma(v)| = \Delta$ and $\slv = \bigcup_{u \in \Gamma(v)} \bl(u)$.
%
Indeed, for any $\conf \in \bigotimes_{u \in \Gamma(v)} \bl(u)$, we have
\[
\{\, \conf_u : u \in \Gamma(v) \,\}
\subseteq \slv
\quad \text{and} \quad
|\{\, \conf_u : u \in \Gamma(v) \,\}| \le |\Gamma(v)| = \Delta.
\]
Hence $\mathcal{P}_v \subseteq \mathcal{P}(\slv_\Delta)$, 
where $\mathcal{P}_v$ denotes the family of local update distributions at $v$ defined in \eqref{eq:bounding-glauber-dynamics-updates}.
Consequently, any grand coupling $f$ constructed on $\mathcal{P}(\slv_\Delta)$ in \Cref{prob:optimal-grand-coupling-2} 
remains a valid grand coupling on $\mathcal{P}_v$ in \Cref{prob:optimal-grand-coupling}.

The simplified formulation in \Cref{prob:optimal-grand-coupling-2} is not only conceptually cleaner, but also captures the key structure underlying previous designs.  
In particular, the core grand couplings employed in the algorithms of Bhandari et al.~\cite{bhandari2020improved} (Algorithm~3, \textsc{Contract}) and Jain et al.~\cite{jain2021perfectly} (Algorithm~3, \textsc{Seeding}) can both be viewed within the framework of \Cref{prob:optimal-grand-coupling-2}.  
Building on this connection, we present a new construction that optimally solves \Cref{prob:optimal-grand-coupling-2}.


\subsection{Optimal design of grand coupling}
\label{sec:optimal-seeding}
In this subsection, we present an optimal construction of the grand coupling for \Cref{prob:optimal-grand-coupling-2}, focusing on the nontrivial case where $\Delta < |\slv| < q$.


For a grand coupling $f:\mathcal{P}\times [0,1] \to [q]$ that is feasible for \Cref{prob:optimal-grand-coupling-2},  
define, for each $U\in[0,1]$,  the bounding set of possible outcomes
$$\mathcal{L}_U \defeq \{\, f(p, U) : p \in \mathcal{P} \,\}.$$
The grand coupling $f$ naturally induces a probability distribution $\{r_k\}_{k\ge 1}$ over the sizes of $\bl_U$, where 
$$r_k = \Pr[{U\sim\uniform[0,1]}]{\,|\mathcal{L}_U| = k\,}.$$

We next present a system of linear constraints on the probabilities $\{r_k\}$ which correspond to the grand couplings that are feasible for \Cref{prob:optimal-grand-coupling-2}.

\begin{lemma}[Linear system for grand couplings]
\label{lem:rk-constraints}
Fix an instance $(\slv,\Delta)$ for \Cref{prob:optimal-grand-coupling-2}, and let $\mathcal{P}=\mathcal{P}(\slv_\Delta)$.  
Let $\{r_k\}_{k=1}^\Delta$ be a probability distribution on $\{1,2,\dots,\Delta\}$ (so that $\sum_{k=1}^\Delta r_k = 1$).  

Consider the system of linear constraints 
\begin{align}
\label{eq:linear-constraint-rk-j}
  \sum_{k=1}^{\Delta} r_k \cdot 
  \frac{\binom{j}{k-1}}{\binom{|\slv|}{k-1}}
   \le \frac{q - |\slv|}{q - j},\qquad \forall 1\le j\le \Delta.
\end{align}
Then the following statements hold:
\begin{enumerate}[label=(\arabic*),wide]
\item \textbf{(Necessity)}
\label{LP-Necessity}
Let $f : \mathcal{P}\times[0,1]\to[q]$ be a grand coupling that is feasible for Problem~\ref{prob:optimal-grand-coupling-2} on
$(\slv,\Delta)$.
The sequence of probabilities $r_k=\Pr[{U\sim \uniform[0,1]}]{\,|\mathcal{L}_U|=k\,}$, for $k=1,\dots,\Delta$, satisfies \eqref{eq:linear-constraint-rk-j}. 
\item \textbf{(Sufficiency)}
\label{LP-Sufficiency}
Conversely, if a probability distribution $\{r_k\}_{k=1}^\Delta$ satisfies \eqref{eq:linear-constraint-rk-j}, then there exists a grand coupling $f : \mathcal{P}\times[0,1]\to[q]$ that is feasible for Problem~\ref{prob:optimal-grand-coupling-2} on
$(\slv,\Delta)$ and realizes the size law
\[
    \Pr[{U\sim \uniform[0,1]}]{\,|\mathcal{L}_U| = k\,} = r_k,    \quad k=1,2,\dots,\Delta.
\]
\end{enumerate}
\end{lemma}

\begin{remark}
\label{remark:Delta-constraint}
In Lemma~\ref{lem:rk-constraints} we restrict attention to grand couplings for which the random bounding list $\mathcal{L}_U$ always has size at most $\Delta$.  This restriction is without loss of generality: the permutation-based ``unconstrained'' coupling from Section~\ref{sec:warm-up} already produces a bounding list of size $\Delta+1$ for every seed. 
Thus, any coupling that sometimes outputs size $\Delta+1$ can be treated separately as the straightforward (unrefined) case. Consequently, when searching for non-straightforward, size-reducing couplings, it suffices to consider laws supported on $\{1,2,\dots,\Delta\}$.
\end{remark}

Since the objective of \Cref{prob:optimal-grand-coupling-2} is to minimize $\mathbb{E}[|\mathcal{L}_U|] = \sum_{k} k \cdot r_k$ where $r_k = \Pr{|\mathcal{L}_U| = k}$,
we can equivalently reformulate \Cref{prob:optimal-grand-coupling-2} as a linear program over $\{r_k\}$ subject to~\eqref{eq:linear-constraint-rk-j}.

\begin{theorem}
\label{thm:lp-characterization}
Fix integers $\Delta \ge 1$ and $q > \Delta$, and let $\slv \subseteq [q]$ satisfy $\Delta < |\slv| < q$.
For each $j,k \in \{1,2,\dots,\Delta\}$, define
\[
z_j(k) \;\defeq\; \frac{\binom{j}{k-1}}{\binom{|\slv|}{k-1}}.
\]
Consider the linear program
\begin{align}
\label{eq:rk-lp-full}
\text{minimize}\quad & \sum_{k=1}^\Delta k\cdot r_k \\
\text{subject to}\quad & \sum_{k=1}^\Delta r_k = 1, \nonumber\\
                      & \sum_{k=1}^\Delta r_k \cdot z_j(k) \le \dfrac{q-|\slv|}{q-j}
                        &&j=1,2,\dots,\Delta, \nonumber\\
                      & r_k \ge 0 &&k=1,2,\dots,\Delta. \nonumber
\end{align}
Let $\mathrm{OPT}_{\mathrm{LP}}$ denote the optimal objective value of \eqref{eq:rk-lp-full}.

Then $\mathrm{OPT}_{\mathrm{LP}}$ equals the minimum expected number of coupling outcomes over all grand couplings $f : \mathcal{P}(\slv_\Delta)\times [0,1] \to [q]$ that are feasible for
Problem~\ref{prob:optimal-grand-coupling-2} on the instance $(\slv,\Delta)$, i.e.,
\[
\mathrm{OPT}_{\mathrm{LP}}=\min_{f:\text{ a grand coupling on }\mathcal{P}(\slv_\Delta)}\mathbb{E}_{U \sim \uniform[0,1]}\!\bigl[|\{f(p, U) : p \in \mathcal{P}\}|\bigr].
\]
\end{theorem}

\Cref{thm:lp-characterization} follows directly from \Cref{lem:rk-constraints}, since \Cref{lem:rk-constraints} provides both necessary and sufficient conditions for any sequence $\{r_k\}_{k=1}^\Delta$ that correspond to a feasible solution to \Cref{prob:optimal-grand-coupling-2}.

We next provide a proof of Lemma~\ref{lem:rk-constraints}, addressing both the necessity and sufficiency of the linear constraints.

\begin{proof}[Proof of \Cref{lem:rk-constraints}]
We split the proof into two parts corresponding to the two statements of the lemma.

\paragraph{\textbf{Proof of \cref{LP-Necessity}}}

Let $f : \mathcal{P}\times[0,1]\to[q]$ be a grand coupling that is feasible for Problem~\ref{prob:optimal-grand-coupling-2} on
$(\slv,\Delta)$.
For each $U \in [0,1]$, recall that
$\mathcal{L}_U \defeq \{\, f(p, U) : p \in \mathcal{P}\,\}$
denote the bounding set of possible outcomes of the grand coupling $f$ for the random seed $U$.  

For convenience, define $\tlv \defeq [q]\setminus \slv.$
Intuitively, $\tlv$ consists of colors that are guaranteed to be unused by neighbors, as they do not appear in the neighborhood bounding list.

Given any bounding list $\mathcal{L}_U$, we can always decompose it into the part contained in $\slv$ and the part contained in $\tlv$, that is,
\[
\mathcal{L}_U = \mathcal{L}^\slv_U \cup \mathcal{L}^\tlv_U, 
\qquad \mathcal{L}^\slv_U \subseteq \slv,\;\; \mathcal{L}^\tlv_U \subseteq \tlv.
\]

We first show that the component $\mathcal{L}^\tlv_U$ must be non-empty for any $U\in[0,1]$.  
By contradiction, suppose $\mathcal{L}^\tlv_U = \emptyset$ for some $U\in[0,1]$.  
Then we have $\mathcal{L}_U = \mathcal{L}^\slv_U \subseteq \slv$ and $|\mathcal{L}_U| \le \Delta$.  
Taking $C = \mathcal{L}_U$, the distribution $\uniform([q]\setminus C) \in \mathcal{P}$ would be a valid input of \Cref{prob:optimal-grand-coupling-2}, but all colors in its support would lie outside the bounding list $\mathcal{L}_U$, violating the correctness of the grand coupling.

Recall that the probability sequence $\{r_k\}_{k=1}^\Delta$ is defined by $r_k=\Pr[{U\sim \uniform[0,1]}]{\,|\mathcal{L}_U|=k\,}$, for $k=1,\dots,\Delta$.
Fix an integer $j \in \{1,2,\dots,\Delta\}$. Let $\mathcal{S}^*_j \defeq \{\, C \subseteq \slv : |C| \le \Delta \,\}$ denotes the collection of all subsets of $\slv$ of size exactly~$j$.

Now consider sampling a neighbor configuration $C \in \mathcal{S}^*_j$ uniformly at random.
Then for any fixed $\mathcal{L}_U$ of size $k$, note that the subset $\mathcal{L}^\slv_U \subseteq \slv$ has size at most $k-1$, since we have already established that $\mathcal{L}^\tlv_U$ is non-empty. 
The probability that a uniformly sampled configuration $C$ contains all colors in $\mathcal{L}^\slv_U$ is therefore lower bounded by
\begin{align}
\label{eq:LP-nec-1}
\Pr[C \sim \uniform \mathcal{S}^*_j]{\mathcal{L}^\slv_U \subseteq C}
= \frac{\binom{j}{|\mathcal{L}^\slv_U|}}{\binom{|\slv|}{|\mathcal{L}^\slv_U|}}
\ge \frac{\binom{j}{k-1}}{\binom{|\slv|}{k-1}}.
\end{align}

On the other hand, 
by marginal correctness of grand coupling we therefore have, for every fixed $C$, let $p_c \defeq \uniform([q]\setminus C)$,
\[
\Pr[{U\sim\uniform[0,1]}]{ f(p_C,U)\in\tlv } \;=\; \frac{q-|\slv|}{q-j}.
\]
Averaging over $C\sim\uniform(\mathcal{S}_j)$ yields the identity
\begin{align}
\label{eq:LP-nec-2}
\Pr[{C\sim\uniform(\mathcal{S}_j),\;U\sim\uniform[0,1]}]{ f(p_C,U)\in\tlv } \;=\; \frac{q-|\slv|}{q-j}.
\end{align}

Finally observe that the event $\mathcal{L}_U^\slv \subseteq C$ implies $f(p_C,U)\in\tlv$: if every color from $\mathcal{L}_U^\slv$ is banned by $C$, then no color in $\slv$ can be produced by $f(p_C,U)$.
Combining \eqref{eq:LP-nec-1}, \eqref{eq:LP-nec-2} and the above observation, using the law of total probability yields
\begin{align*}
\frac{q-|\slv|}{q-j}
&= \Pr[{C,U}]{ f(p_C,U)\in\tlv } \\
&= \sum_{k=1}^\Delta \Pr[{U}]{\,|\mathcal{L}_U|=k\,}
      \cdot \Pr[{C,U}]{ f(p_C,U)\in\tlv \;\big|\; |\mathcal{L}_U|=k } \\
&\ge \sum_{k=1}^\Delta \Pr[{U}]{\,|\mathcal{L}_U|=k\,}
      \cdot \Pr[{C,U}]{ \mathcal{L}_U^\slv \subseteq C \;\big|\; |\mathcal{L}_U|=k } \\
&\ge \sum_{k=1}^\Delta r_k \cdot \frac{\binom{j}{k-1}}{\binom{|\slv|}{k-1}}.
\end{align*}
Since the above inequality holds for each fixed $j\in\{1,2,\dots,\Delta\}$,  
we conclude that the collection $\{r_k\}_{k=1}^\Delta$ satisfies exactly the family of inequalities specified in~\eqref{eq:linear-constraint-rk-j}.


\paragraph{\textbf{Proof of \cref{LP-Sufficiency}}}

Recall that $\mathcal{S}_\Delta \defeq \{\, C \subseteq \slv : |C| \le \Delta \,\}$ and $\tlv \defeq [q]\setminus\slv$.  
Let $\{r_k\}_{k=1}^\Delta$ be a probability distribution satisfies \eqref{eq:linear-constraint-rk-j}, and for each $C\in \mathcal{S}_\Delta$, define
\[
P_C \;=\; \sum_{k=1}^{\Delta} r_k \cdot \frac{\binom{|C|}{k-1}}{\binom{|\slv|}{k-1}},
\qquad
Q_C \;=\; \frac{q-|\slv|}{\,q-|C|\,}.
\]
By \eqref{eq:linear-constraint-rk-j} we have $P_C \le Q_C$ for every such $C$. (Note that $|\slv|>\Delta$ in our setting, so $Q_C<1$.)

We now construct a grand coupling $f:\mathcal{P}\times[0,1]\to [q]$ that realizes the prescribed distribution $\{r_k\}$.  
Concretely, for each seed $U\in[0,1]$ we deterministically extract a   where
\begin{itemize}
  \item $K\in\{1,\dots,\Delta\}$ is drawn with $\Pr{K=k}=r_k$;
  \item $\pi$ is a uniformly random permutation of $\slv$;
  \item $c_0$ is drawn uniformly from $\tlv$;
  \item $U'\sim\uniform[0,1]$ is independent of the previous choices.
\end{itemize}
As usual, such independent randomness can be encoded measurably into a single uniform seed $U$.

Given an input distribution $p=\uniform([q]\setminus C)\in\mathcal{P}$ (where $C\in \mathcal{S}_\Delta$), and given the internal randomness $(K,\pi,c_0,U')$ obtained from $U$, let $A=\{\pi_1,\dots,\pi_{K-1}\}$ denote the prefix of $\pi$ of length $K-1$. Define
\[
\alpha_C \;=\; \frac{1-Q_C}{1-P_C} \in (0,1].
\]
We then set
\begin{align}
\label{eq:bounding-chain-opt}
f\bigl(\uniform([q]\setminus C),U\bigr)=
\begin{cases}
\text{the first element of $A$ that lies in $\slv\setminus C$,} & \text{if } A \not\subseteq C \text{ and } U' < \alpha_C,\\[4pt]
c_0, & \text{otherwise.}
\end{cases}
\end{align}

We verify that $f$ is a valid grand coupling and that it realizes the law $\Pr{|\mathcal{L}_U|=k}=r_k$.

First we check marginal correctness. 
Let $X$ denote the random variable corresponding to the output of $f$ above.
Condition on $C$. By construction,
\[
\Pr{A\subseteq C} \;=\; \sum_{k=1}^\Delta r_k \cdot \frac{\binom{|C|}{k-1}}{\binom{|\slv|}{k-1}} \;=\; P_C,
\]
hence $\Pr{A\not\subseteq C}=1-P_C$. The algorithm outputs a color in $\tlv$ exactly when either $A\subseteq C$ (in which case we always output $c_0$) or when $A\not\subseteq C$ but $U'\ge\alpha_C$. Therefore
\begin{align*}
\Pr{X\in \tlv}
&= \Pr{A\subseteq C} + \Pr{A\not\subseteq C}\cdot(1-\alpha_C) \\
&= P_C + (1-P_C)\Bigl(1 - \frac{1-Q_C}{1-P_C}\Bigr) \;=\; Q_C.
\end{align*}
When the output falls in $\tlv$ it equals the independent uniform color $c_0$, so conditional on this event the distribution over $\tlv$ is uniform; consequently each color $c\in\tlv$ is output with marginal probability
\[
\Pr{X=c} \;=\; \frac{Q_C}{|\tlv|} \;=\; \frac{1}{q-|C|}.
\]

It remains to check the colors in $\slv\setminus C$. The event that the output lies in $\slv\setminus C$ occurs exactly when $A\not\subseteq C$ and $U'<\alpha_C$. Conditioned on $A\not\subseteq C$ and $K=k$, the decoder returns the first element of $A$ that lies in $\slv\setminus C$. By symmetry of $\pi$, for fixed $k$ this "first-not-in-$C$" element is uniformly distributed over $\slv\setminus C$ (whenever $A\not\subseteq C$). Therefore each color $c\in\slv\setminus C$ has conditional probability $1/(|\slv|-|C|)$ given that the output is in $\slv\setminus C$. Since $\Pr{X\in\slv\setminus C} = 1-Q_C$,
we deduce that for every $c\in\slv\setminus C$
\[
\Pr{X=c} = (1-Q_C)\cdot\frac{1}{|\slv|-|C|} = \frac{1}{q-|C|}.
\]
Combining with the uniformity on $\tlv$, we conclude that the distribution of $X$ equals $\uniform([q]\setminus C)$, as required. Hence $f$ is a valid grand coupling.

Finally we verify that the coupling indeed realizes the prescribed size distribution.
Recall that the quadruple $(K,\pi,c_0,U')$ is deterministically decoded from the seed $U\sim\uniform[0,1]$,
and $A=\{\pi_1,\dots,\pi_{K-1}\}$ denote the prefix of $\pi$ of length $K-1$.
By the construction of $f$ (see~\eqref{eq:bounding-chain-opt}),
for every choice of set $C\subseteq\slv$ with $|C|\le\Delta$ the output of the decoder is either $c_0$ or the first element of $A$ that does not lie in $C$.
Consequently, the image of $\mathcal{P}$ under $f(\cdot,U)$ is contained in $A\cup\{c_0\}$.
Moreover, it can be easily verified that for each color in $A\cup\{c_0\}$ there exists some set $C$ for which the decoder outputs that color, so every element of $A\cup\{c_0\}$ is actually attainable.
Therefore for this fixed seed $U$ we have
\[
\mathcal{L}_U=\{\, f(p,U) : p\in\mathcal{P} \,\} = A\cup\{c_0\},
\]
and hence $|\mathcal{L}_U| = |A|+1 = K$.

Since $K$ was drawn according to $\Pr{K=k}=r_k$, it follows that
$\Pr{\,|\mathcal{L}_U|=k\,} = \Pr{K=k} = r_k$
for every $k=1,2,\dots,\Delta$, as required.

This completes the construction and the proof of \Cref{LP-Sufficiency}.
\end{proof}

\paragraph{\textbf{Solving the linear program in \eqref{eq:rk-lp-full}}}
We now discuss how to approach the solution of the linear program \eqref{eq:rk-lp-full}, 
which characterizes the optimal grand coupling via the distribution $\{r_k\}_{k=1}^\Delta$.  

We begin by examining a relaxed version of the program that retains only the constraint corresponding to $j=\Delta$ in \eqref{eq:rk-lp-full}.
Set
\[
z(k)\;=\; z_\Delta(k) \;=\; \frac{\binom{\Delta}{k-1}}{\binom{|\slv|}{k-1}},
\qquad
w \;=\; \frac{q-|\slv|}{q-\Delta},
\]
and consider the reduced linear program
\begin{align}
\label{eq:rk-lp}
\text{minimize}\quad & \sum_{k=1}^\Delta k\cdot r_k \nonumber\\
\text{subject to}\quad & \sum_{k=1}^\Delta r_k = 1, \\ 
                      & \sum_{k=1}^\Delta r_k \cdot z(k) \le w, \nonumber\\
                      & r_k \ge 0 \quad (k=1,\dots,\Delta). \nonumber
\end{align}
The function $z(k)$ defined above is convex in $k$ for fixed $\Delta$ and $|\slv| > \Delta$. By standard extremal properties of linear programs with a single convex moment constraint (see, e.g., \cite{hardy1952inequalities} and related moment–LP arguments), any optimal solution of \eqref{eq:rk-lp} is supported on at most two consecutive values of $k$. In other words, there exists an index $1<i\le \Delta$ such that the optimal distribution $\{r_k\}$ has $r_k=0$ for all $k\not\in\{i-1,i\}$. Consequently the optimal $\{r_k\}$ (when it exists) can be constructed as follows. Let $i\in\{2,\dots,\Delta\}$ be the smallest index such that $z(i)\le w$. If such an $i$ exists, set
\begin{equation}
\label{eq:rk-opt}
r_{i-1} \;=\; \frac{w - z(i)}{\,z(i-1) - z(i)\,}, \qquad
r_{i}   \;=\; \frac{z(i-1) - w}{\,z(i-1) - z(i)\,},
\end{equation}
with all other $r_k=0$.  By construction $r_{i-1},r_i\ge0$, $r_{i-1}+r_i=1$,
and this choice satisfies $\sum_k r_k \cdot z(k)=w$, hence is optimal for
\eqref{eq:rk-lp}.

\begin{remark}
\label{remark:lp}
The program \eqref{eq:rk-lp} is a relaxation of the original LP
\eqref{eq:rk-lp-full}. In our subsequent analysis we work with the relaxed program \eqref{eq:rk-lp} because it admits the explicit solution \eqref{eq:rk-opt}. In typical parameter regimes of interest
(e.g.\ when $q = 2\Delta$), one can verify that satisfying the $j=\Delta$ constraint (i.e.\ the single constraint used in the relaxed program) generally implies the remaining constraints for $j=1,2,\dots,\Delta-1$.
Thus the relaxation is non-pathological for the regimes we care about, and the explicit solution \eqref{eq:rk-opt} yields the optimal (or near-optimal) distribution of $|\mathcal{L}_U|$ in those settings.
\end{remark}

Combining the explicit probabilities in~\eqref{eq:rk-opt} with the constructive proof in Lemma~\ref{lem:rk-constraints} thus yields an explicit optimal grand coupling for Problem~\ref{prob:optimal-grand-coupling-2} under the stated assumptions.
\color{black}

\subsection{Constructions of grand couplings}
\label{sec:grand-coupling-applications}

In this subsection, we present several concrete constructions of grand couplings that serve as building blocks for the CFTP algorithm stated in \Cref{thm:main}.

\subsubsection{Grand coupling: \seeding}
The grand coupling \emph{\seeding} refers to the optimal construction introduced in \Cref{sec:optimal-seeding}.  
This coupling is the key to our improvement of the coloring threshold, reducing from the previous bound $q > (8/3 + o(1))\,\Delta$~\cite{jain2021perfectly} to the nearly optimal bound $q > (2.5 + o(1))\,\Delta$.

\begin{lemma}[Seeding coupling]
\label{lem:seeding}
Fix any vertex $v \in V$ and neighborhood bounding lists $(\bl(u))_{u \in \Gamma(v)}$. 
Let $\mathcal{P}_v$ denote the set of all possible local update distributions at $v$ in the Glauber dynamics on $q$-colorings, given the bounding lists $(\bl(u))_{u \in \Gamma(v)}$ on $v$’s neighbors. Formally,
\begin{align}\label{eq:bounding-glauber-dynamics-updates-2}
     \mathcal{P}_v = \left\{\uniform([q] \setminus \conf): \conf \in \bigotimes_{u\in\Gamma(v)}\bl(u)\right\}
 \end{align}
%
Assume 
\[
    q \ge \tfrac{7}{3}\,\Delta
    \quad\text{and}\quad
    |\slv| \le 2\Delta,
    \quad\text{where }\slv \defeq \bigcup_{u \in \Gamma(v)} \bl(u).
\]
Then there exists a grand coupling $\seeding_v : \mathcal{P}_v \times [0,1] \to [q]$ satisfying:
\begin{enumerate}[label=(\alph*), ref=\alph*]
    \item \label{seeding-size} (\textbf{Bounding list reduction }) 
    For random seed $U \in [0,1]$, let $\bl_U(v) = \{ \seeding_v(p, U) : p \in \mathcal{P}_v \}$ denote the set of possible coupling outcomes. Then $|\bl_U(v)| \in \{2,3\}$, and
    $$
        \mathbb{E}_{U \sim \uniform[0,1]}\!\bigl[\,|\bl_U(v)|\,\bigr] \le 2 + \frac{(|\slv| + \Delta - q)(|\slv|-1)}{(q-\Delta)\,\Delta}.
    $$
    \item \label{seeding-efficiency} (\textbf{Efficiency}) 
    The grand coupling $\seeding_v$ can be implemented so that each update of the bounding list can be performed in expected time $O(\Delta \log \Delta)$.
\end{enumerate}
\end{lemma}

\begin{proof}
Recall that solving Problem~\ref{prob:optimal-grand-coupling-2} for the input $(\slv,\Delta)$ yields a grand coupling that is valid for any neighborhood bounding lists with union equal to $\slv$.  Hence it suffices to produce an admissible probability law $\{r_k\}_{k=1}^\Delta$ satisfying the feasibility constraints of Lemma~\ref{lem:rk-constraints}. The lemma then guarantees existence of a corresponding grand coupling realizing
$\Pr{\,|\mathcal{L}_U(v)| = k\,}=r_k$ for every $k$.

We take $\{r_k\}$ to be the two-point law obtained from the explicit
formula \eqref{eq:rk-opt} with index $i=3$.  Concretely define
\[
r_3 \;=\; \frac{(|\slv| + \Delta - q)(|\slv|-1)}{(q-\Delta)\,\Delta},
\qquad
r_2 \;=\; 1-r_3,
\qquad
r_k \;=\; 0 \quad (k\notin\{2,3\}).
\]

Under the parameter restrictions of the lemma ($q\ge 7/3\,\Delta$ and $|\slv|\le 2\Delta$) one checks routinely that
$r_2,r_3\in[0,1]$.  
It can also be easily verified that the above choice of
$\{r_k\}$ satisfies the linear constraint~\eqref{eq:linear-constraint-rk-j} throughout this regime, as earlier mentioned in Remark ~\ref{remark:lp}.

By Lemma~\ref{lem:rk-constraints}, there exists a valid grand coupling $f$ whose explicit construction was given in the proof of that lemma.  
For $U \sim \uniform[0,1]$, this coupling satisfies $\Pr{\,|\mathcal{L}_U(v)| = k\,} = r_k$ for every $k$.  
We take precisely this constructed coupling as our $\seeding_v$.

\medskip\noindent\textbf{Proof of \cref{seeding-size}.}
Because $\Pr{\,|\mathcal{L}_U(v)|=k\,}=r_k$ and only $r_2,r_3$ are nonzero,
the size $|\mathcal{L}_U(v)|$ takes values in $\{2,3\}$ and its
expectation equals
\[
\mathbb{E}\bigl[\,|\mathcal{L}_U(v)|\,\bigr]
= 2\cdot r_2 + 3\cdot r_3
= 2 + r_3
= 2 + \frac{(|\slv| + \Delta - q)(|\slv|-1)}{(q-\Delta)\,\Delta},
\]
which establishes the claimed bound.

\medskip\noindent\textbf{Proof of \cref{seeding-efficiency}.}
Since $|\slv| \le 2\Delta$, constructing $\slv = \bigcup_{u \in \Gamma(v)} \bl(u)$ requires $O(\Delta \log \Delta)$ time.  
Within $\seeding_v$, apart from generating a random permutation of $\slv$ in $O(\Delta)$ time, all encoding and decoding operations take constant time.  
Hence, the overall computational complexity is $O(\Delta \log \Delta)$.
\end{proof}

\subsubsection{Grand Coupling: $\compress$ and $\disjoint$}

\label{sec:grand-coupling-compress}

Next, we present two canonical constructions of grand couplings, referred to as \emph{\compress} and \emph{\disjoint}, which serve as fundamental components in the CFTP algorithm discussed later.  

The \compress{} and \disjoint{} algorithms were originally introduced in~\cite{bhandari2020improved} and~\cite{jain2021perfectly}, respectively. Both algorithms can be naturally interpreted as specific instances of grand couplings via the bounding chains within our framework.

In particular, the \compress{} coupling can be viewed as a refinement of Huber's permutation-based grand coupling~\cite{huber1998exact} discussed in~\Cref{sec:warm-up}, optimized for the bounding-chain setting.

\begin{lemma}[Compress coupling]
\label{lem:compress}
Fix any vertex $v \in V$ and neighborhood bounding lists $(\bl(u))_{u \in \Gamma(v)}$. 
Let $\mathcal{P}_v$ be defined as \eqref{eq:bounding-glauber-dynamics-updates-2}.
Assume $q \ge \Delta + 1$. 
For any fixed subset $A \subseteq [q]$ with $|A| = \Delta$, called a \emph{reference color set}, there exists a grand coupling $\compress^A_v:\mathcal{P}_v\times[0,1]\to[q]$ satisfying:
\begin{enumerate}[label=(\alph*), ref=\alph*]
    \item \label{compress-size} (\textbf{Bounding property}) 
    For random seed $U \sim\uniform[0,1]$, the set of possible outcomes is
    \[
        \set{\compress^A_v(p, U) : p \in \mathcal{P}_v} = A \cup \{ c \},
    \]
    where $c$ is distributed uniformly in $[q]\setminus A$ (with the randomness induced by $U$).
    \item \label{compress-efficiency} (\textbf{Efficiency}) 
    The grand coupling $\compress^A_v$ can be implemented so that each update of the bounding list can be performed in expected time $O(\Delta \log \Delta)$.
\end{enumerate}
\end{lemma}

\begin{proof}
We encode the internal randomness of the $\compress$ procedure into a single random seed $U\in[0,1]$.  
Concretely, for each seed~$U$ we deterministically extract a triple $(\pi,x',U')$ where
\begin{itemize}
    \item $\pi$ is a uniformly random permutation of $A$,
    \item $x'$ is a uniformly random element of $[q]\setminus A$,
    \item $U'$ is an independent uniform draw from $[0,1]$.
\end{itemize}

For each feasible neighbor configuration $\conf\in\mathsf{Conf}(v)$—with induced forbidden set $\Gamma_\conf(v)$ and corresponding distribution $p_\conf = \uniform([q]\setminus\Gamma_\conf(v))\in \mathcal{P}_v$—define
\[
\compress^A_v(p_\conf,U) =
\begin{cases}
x', & \text{if } x' \notin \Gamma_\conf(v)\ \text{and}\ U' \le \dfrac{q-\Delta}{\,q-|\Gamma_\conf(v)|\,},\\[6pt]
\text{the first } y \text{ in } \pi \text{ with } y \notin \Gamma_\conf(v), & \text{otherwise.}
\end{cases}
\]

\noindent\textbf{Grand coupling property (correctness).}  
Fix a configuration~$\conf$ and let $\Gamma=\Gamma_\conf(v)$.  
Conditional on $(\pi,x')$ (or marginalizing over them), the above rule returns~$x'$ with probability
\[
  \Pr{x'\notin\Gamma} \cdot \Pr{U' \le \tfrac{q-\Delta}{q-|\Gamma|}}
  = \frac{1}{q-\Delta}\cdot \frac{q-\Delta}{q-|\Gamma|} = \frac{1}{q-|\Gamma|}.
\]
If the first branch is not taken, the decoder returns the first element of~$\pi$ outside~$\Gamma$; since $\pi$ is a uniform permutation of~$A$, each color in $[q]\setminus\Gamma$ (that lies within~$A$) is equally likely to be the first available, and its probability is again $1/(q-|\Gamma|)$.  
Hence marginally $\compress^A_v(p_\conf,U)\sim \uniform([q]\setminus\Gamma)$, verifying that $\compress^A_v$ is a valid grand coupling in the sense of Definition~\ref{def:grand-coupling-on-bounding-chain}.

\medskip
\noindent\textbf{Proof of \cref{compress-size}.}  
For any fixed seed (equivalently fixed $(\pi,x',U')$), the image of $\mathcal{P}_v$ under $\compress^A_v(\cdot,U)$ is contained in $A\cup\{x'\}$.  
Conversely, $x'$ itself is attainable (for those configurations with $x'\notin\Gamma$ and satisfying the $U'$-condition), and every element of~$A$ is attainable (for configurations that force the decoder to use~$\pi$).  
Thus for each seed the image equals $A\cup\{c\}$ for some $c\in[q]\setminus A$, and the distribution of $c$ over seeds is uniform on $[q]\setminus A$.

\medskip
\noindent\textbf{Proof of \cref{compress-efficiency}.}  
Decoding the seed (recovering $(\pi,x',U')$) and forming the predicted set $A\cup\{x'\}$ both take $O(\Delta)$ time.  
Given a concrete configuration~$\conf$, evaluating $\compress^A_v(p_\conf,U)$ requires scanning~$\pi$ until the first element not in~$\Gamma_\conf(v)$, which can be implemented in expected $O(\Delta \log \Delta)$ time using standard hashing or balanced lookup structures.  
This yields the claimed efficiency bound.

This completes the proof of Lemma~\ref{lem:compress}.
\end{proof}

The \disjoint{} coupling corresponds to a key step in~\cite{jain2021perfectly}, enabling the reduction of all bounding list sizes from 2 to 1 in the regime where $q>2.5\Delta$.

\begin{lemma}(Disjoint coupling \cite[Lemma 4.3]{jain2021perfectly})
\label{lem:disjoint}
Fix any vertex $v \in V$ and neighborhood bounding lists $(\bl(u))_{u \in \Gamma(v)}$. 
Let $\mathcal{P}_v$ be defined as \eqref{eq:bounding-glauber-dynamics-updates-2}.
Define 
\begin{align*}
    \slv &\defeq \bigcup_{u \in \Gamma(v)} \bl(u),\\
    \mathcal{Q} &\defeq \bigcup_{\substack{u \in \Gamma(v):|\mathcal{L}(u)|=1}} \mathcal{L}(u),\\
    \mathcal{D} &\defeq \bigcup_{\substack{u \in \Gamma(v):\\\forall w\in\Gamma(v), \mathcal{L}(u)\cup\mathcal{L}(w)=\varnothing}} \mathcal{L}(u).
\end{align*}
%
Assume $q \ge 2.5 \Delta$ and $|\mathcal{L}(u)| \leq 2$ for all $u\in\Gamma(v)$.
Then, there exists a grand coupling $\disjoint:\mathcal{P}_v\times[0,1]\to[q]$ satisfying:
    \begin{enumerate}[label=(\alph*), ref=\alph*]
    \item \label{disjoint-size} (\textbf{Bounding list reduction}) 
    For random seed $x \in [0,1]$, let $\bl_U(v) = \{ \disjoint_v(p, x) : p \in \mathcal{P}_v \}$ denote the set of possible coupling outcomes. Then $|\bl_U(v)|\in\{1,2\}$, and 
    $$
        \Pr[U \sim \uniform\mathrm{[}0,1\mathrm{]}]{|\bl_U(v)|=1} \ge 1-\frac{|\slv|-|\qlv|}{q-\Delta}+\frac{|\dlv|/2}{q-|\qlv|-|\dlv|/2}.
    $$
    \item \label{disjoint-efficiency} (\textbf{Efficiency}) 
    The grand coupling  $\disjoint_v$ can be implemented so that each update of the bounding list can be performed in expected time $O(\Delta \log \Delta)$.
\end{enumerate}
\end{lemma}

\begin{remark}
The disjoint coupling builds on Jerrum’s pairing coupling analysis of Glauber dynamics~\cite{jerrum1995very}, refined under the bounding-list framework.
Its key idea is to \emph{perfectly couple} the pair of colors within the same independent bounding set, making them to be regarded as single color, effectively reducing the neighborhood size to $E = |\slv| - \tfrac{1}{2}|\dlv|$.
While~\cite{bhandari2020improved} treated $E = |\slv|$ and thus required $q > 3\Delta$, the refinement of~\cite{jain2021perfectly} showed that when $|\mathcal{L}(u)| \le 2$, one has $E \le 1.5\Delta$, so $q > 2.5\Delta$ ensures contraction.
Our \cref{thm:no-coupling-2} show that this $2.5\Delta$ threshold is tight, making the disjoint coupling optimal under these assumptions.
\end{remark}
\section{CFTP for \texorpdfstring{$q$}{q}-Colorings with \texorpdfstring{$q>2.5\Delta$}{q > 2.5Δ}} 
\label{sec:algorithm}

In this section, we present the construction of a perfect sampler for proper $q$-colorings within the classical
\emph{Coupling From The Past (CFTP)} framework, implemented using the grand couplings developed in \Cref{sec:coupling-design}. 
We show that this construction gives an efficient CFTP algorithm for exact sampling from the uniform distribution over proper $q$-colorings, 
under the condition $\Delta>(2.5+\eta)\Delta$ where $\eta=2\sqrt{(\log\Delta+1)/\Delta}=o(1)$, thus completing the proof of \Cref{thm:main}.

\subsection{Algorithm overview} 
\label{subsec:algorithm-overview}
Let $\Omega$ be a finite configuration space endowed with the uniform measure $\mu_\Omega$,
and let $F : \Omega\times[0,1]\to\Omega$ be a randomized \emph{update operator} driven by a random seed $U\in[0,1]$.
Associated with $F$ is a \emph{coalescence predicate} $\Phi_F : [0,1]\to\Omega\cup\{\bot\}$, 
which detects whether the update $F(\cdot,U)$ maps all initial configurations in $\Omega$ to a common state.  
Specifically, $\Phi_F(U)=\conf^*\neq\bot$ indicates that all trajectories under $F(\cdot,U)$ coalesce at the same configuration $\conf^*\in\Omega$.

\begin{algorithm}[H]
\caption{CFTP Sampler} 
\label{alg:blocked-cftp-sampler}
\SetKwInput{KwAssertion}{Assertion}
\KwIn{Random seed $U\sim\uniform[0,1]$, interpreted as i.i.d.\ $U_{-1},U_{-2},\ldots\sim\uniform[0,1]$; update operator $F:\Omega\times[0,1]\to\Omega$; coalescence predicate $\Phi_F:[0,1]\to\Omega\cup\{\bot\}$.}
\KwOut{A configuration $\conf^{*}$ from $\Omega$.}
    \For{$t=1,2,\cdots$}{
        Generate an independent $U_{-t}\sim\uniform[0,1]$ using the random seed $U$\;
        \If{$\Phi_F(U_{-t}) \neq \bot$}{
            \Return $\conf^* = 
              F(\cdot,U_{-1}) \circ F(\cdot,U_{-2}) \circ \cdots \circ 
              F(\cdot,U_{-t+1})\circ \Phi_F(U_{-t})$
        }
    }
\end{algorithm}

\begin{lemma} 
\label{lemma:blocked-cftp-sampler} 
Let $F : \Omega\times[0,1] \to \Omega$ and $\Phi_F:[0,1]\to \Omega\cup\{\bot\}$ satisfy:
\begin{enumerate}[label=(\alph*), ref=\alph*]
    \item \label{grand-a} For independent $\conf \sim \mu_\Omega$ and $U\sim\uniform[0,1]$, the random variable $F(\conf, U)$ is distributed according to $\mu_\Omega$.
    \item \label{grand-b} If $\Phi_F(U) \neq \bot$, then $F(\cdot, U)$ is constant over $\Omega$, and $\Phi_F(U)$ equals its unique image.
    \item \label{grand-c} $\Pr[U\sim\uniform{[}0,1{]}]{\Phi_F(U) \neq \bot} \ge 1/2$.
    \item \label{grand-d} 
        The predicate $\Phi_F(U)$ can be evaluated in expected time $T_1$, and $F(\omega,U)$ in expected time $T_2$.
\end{enumerate}
Then \Cref{alg:blocked-cftp-sampler} terminates in expected time $O(T_1 + T_2)$ 
and outputs an exact sample from $\mu_\Omega$, with exponentially decaying tail bounds on its runtime.
\end{lemma}

\begin{proof}
Let $\conf_0 \sim \mu_\Omega$ be a initial configuration drawn independently to $U$, and let
\[
    \conf_t = F(\cdot, U_{-1}) \circ F(\cdot, U_{-2}) \circ \cdots \circ F(\cdot, U_{-t})(\conf_0)
\]
denote the configuration obtained after applying $t$ grand couplings from the past. By condition~(\ref{grand-a}), the distribution of $\conf_t$ remains uniform over $\Omega$ for all $t$.

Define $\conf_t^*$ to be the output of Algorithm~\ref{alg:blocked-cftp-sampler} truncated at step $t$ (and $\omega^*_t=\bot$ if algorithm does not halt).
Let $\mu_{\conf_t}$ and $\mu_{\conf_t^*}$ denote the distributions of $\conf_t$ and $\conf_t^*$, respectively. 
Under conditions~(\ref{grand-b}) and~(\ref{grand-c}), 
$$
    \DTV{\mu_{\conf_t}}{\mu_{\conf_t^*}}
    = \Pr[U]{\conf_t \neq \conf_t^*}
    = \Pr[U]{\conf_t^*=\bot}
    = \Pr[U]{\forall i \le t : \Phi_F(U_{-i}) = \bot}
    \le 2^{-t},
$$
by the independence of the seeds $U_{-1}, U_{-2}, \ldots, U_{-t}$.

Taking the limit as $t \to \infty$ yields
$$
    \DTV{\mu_\Omega}{\mu_{\conf^*}}
    = \lim_{t \to \infty} \DTV{\mu_{\conf_t}}{\mu_{\conf_t^*}} = 0,
$$
so the output $\conf^*$ is distributed exactly according to $\mu_\Omega$.

Since each iteration halts independently with probability at least $p\ge 1/2$,
the number of iterations follows a geometric distribution with mean $O(1)$.
Therefore, the total expected runtime is $O(T_1 + T_2)$,
and the probability that the algorithm exceeds $t$ iterations decays exponentially.
\end{proof}

The above lemma provides a general CFTP framework applicable to any update operator satisfying conditions~(\ref{grand-a})–(\ref{grand-d}).
We next instantiate this framework for the Glauber dynamics on $q$-colorings.
\begin{remark}
\label{remark:blocked-grand-coupling-glauber}
In our setting, the randomized update operator $F$ in Lemma~\ref{lemma:blocked-cftp-sampler} 
is realized as a finite composition of single-site grand couplings
for the Glauber dynamics (cf.~\Cref{def:grand-coupling-on-bounding-chain}).
Because the uniform measure $\mu_\Omega$ is stationary for the Glauber dynamics,
condition~(\ref{grand-a}) is automatically satisfied.
Moreover, with the bounding chains, the corresponding coalescence predicate $\Phi_F$
inherently satisfies condition~(\ref{grand-b}). 
Consequently, to prove \Cref{thm:main},
it suffices to construct an explicit family of such grand couplings
that ensure the coalescence probability in~(\ref{grand-c}) and the efficiency guarantees in~(\ref{grand-d}).
\end{remark}

\subsection{Global grand couplings from local updates} 
\label{subsec:construct-grand-coupling}

We now explicitly construct the global grand coupling 
$F^* = F(\cdot, U)$ on the configuration space $\Omega$
and its associated coalescence predicate $\Phi_F(U)$,
by composing local grand couplings corresponding to single-site updates.

\subsubsection{Representing global coupling via composition of local couplings.}

Recall from \Cref{sec:grand-coupling-applications} and \Cref{def:grand-coupling-on-bounding-chain} that for each vertex $v\in V$, we have defined several types of local grand couplings $f_v:\mathcal{P}_v\times[0,1]\to[q]$ corresponding to single-site updates in Glauber dynamics, 
including \seeding{}, \compress{}, and \disjoint{}.
We now describe how these local updates are composed into a global coupling according to a carefully designed schedule.

Throughout this construction, the algorithm has access to a global random seed $U\sim\uniform[0,1]$, 
where each bit of $U$ is sampled i.i.d.\ uniformly so that all randomness is derived from this single source.  

In addition, the construction maintains two global data structures:
\begin{itemize}
    \item A \textbf{bounding list} $\mathcal{L}=(\mathcal{L}(v))_{v\in V}$, 
    where each $\mathcal{L}(v)\subseteq[q]$ records the current set of possible color at vertex~$v$;
    \item A \textbf{composition list} $F^*$, which stores the sequence of local updates used to construct the overall mapping $F(\cdot,U)$.
    Each entry of $F^*$ takes the form $(v,f_v,U_v)$, where $v$ is the vertex being updated, $f_v$ is the local grand coupling, and $U_v\in[0,1]$ is the independent random seed used in that update (generated from the global seed $U$).
\end{itemize}
The following subroutine \textsc{Update}$(v,f)$ modifies the bounding list $\mathcal{L}$ at $v$ and appends a new entry to the composition list $F^*$, given a vertex $v\in V$ and a local grand coupling rule $f=(f_u)_{u\in V}$ specifying at each vertex $u$ a local update function $f_u:\mathcal{P}_u\times[0,1]\to[q]$, as in \Cref{def:grand-coupling-on-bounding-chain}.

\begin{algorithm}[H]
\caption{\textsc{Update}$(v, f)$}
\label{alg:update}
\SetKwInput{KwGlobal}{Global Variable}
\KwIn{Vertex $v\in V$; local grand coupling rule $f=(f_u)_{u\in V}$.}
\KwGlobal{Bounding list $\mathcal{L}$; composition list $F^*$; global random seed $U\sim\uniform[0,1]$.}

Generate an independent $U' \sim \uniform[0,1]$ using the random seed $U$\;
$\mathcal{L}(v) \gets \{f_v(p,U'): p\in\mathcal{P}_v\}$\;
Append tuple $(v,f,U')$ to the end of $F^*$\;
\end{algorithm}

Each call to \textsc{Update} modifies only the bounding list $\mathcal{L}(v)$ at vertex~$v$,
and appends one new record $(v,f_v,U_v)$ representing the local update $f_v(\cdot, U_v)$ to $F^*$.  
Intuitively, each such tuple in $F^*$ specifies a single atomic update applied to the current bounding chain.

The resulting global coupling $F^*:\Omega\to\Omega$ is the composition of all local updates in the order they appear in the list.
Given an initial configuration $\conf\in\Omega$, its image under $F^*$ can be evaluated as follows.

\begin{algorithm}[H]
a\caption{Evaluation of $F^*(\conf)$ for initial configuration $\conf$}
\label{alg:evaluate-F}
\KwIn{Configuration $\conf \in \Omega$; composition list $F^* = [(v_1,f_1,U_1),\dots,(v_T,f_T,U_T)]$.}
\KwOut{Updated configuration $F^*(\conf)$.}

\For{$i=1$ \KwTo $T$}{
    $\conf_{v_i} \gets f_{v_i}(p_{v_i}^\conf, U_i)$, where $p_{v_i}^\conf=\uniform([q]\setminus\conf_{\Gamma(v_i)})$ is the Glauber dynamics update\;
}
\Return $\conf$\;
\end{algorithm}




\subsubsection{Neighborhood maintenance via \textsc{CleanUp}.}
We define an auxiliary routine \textsc{CleanUp}$(v,P)$ that applies a sequence of
$\compress$-type updates around $v$ relative to a reference color set $A\subseteq[q]$ built
from the bounding lists of a preserved set $P\subseteq V$. Intuitively, this
aligns the neighborhoods not currently being updated so that subsequent
\textsc{Update} calls meet the required invariants.

\begin{algorithm}[H]
\caption{\textsc{CleanUp}$(v,P)$}
\label{alg:cleanup}
\SetKwInput{KwGlobal}{Global Variable}
\KwIn{Vertex $v$ to be cleaned; subset $P \subseteq V$ of vertices to preserve.}
\KwGlobal{Bounding list $\mathcal{L}$; composition list $F^*$; global random seed $U\sim\uniform[0,1]$.}

$A\gets \textsc{Greedy}(\bl_P=(\bl(w))_{w\in P})$\;
\ForEach{$w \in \Gamma(v)\setminus P$}{
    \textsc{Update}$(w, \compress^{A})$\;
}
\end{algorithm}

\begin{remark}[The \textsc{Greedy} subroutine in \Cref{alg:cleanup}]
The reference set $A$ used in \textsc{CleanUp} is chosen greedily 
based on the current bounding information $\bl_P$:
the idea is to include as many colors as possible from the target color sets 
until reaching size $\Delta$; and if the total size of the targets is less than $\Delta$, 
we include all of them and fill the remaining slots with arbitrary colors.  
This ensures $|A|=\Delta$ while maximizing coverage over the currently active neighborhood.  

The specific priority rule of this greedy selection may vary across different
stages of the algorithm, as the construction of $A$ will be adapted to the
particular coupling objective of each stage.
The detailed forms of these greedy strategies will be provided later in the
proofs corresponding to each phase of the algorithm.
\end{remark}

\begin{remark}
\label{remark:adaptive-A}
It is important to note that while the vertex $v$ in each 
\textsc{update} or \textsc{CleanUp} call 
must be selected independently of the current configuration, 
the auxiliary set $A$ \emph{may} depend adaptively on the current state 
of the bounding chain. 
The reason is that choosing $A$ only determines 
which specific member of the family of grand couplings 
$\{\compress^{A}:A\subseteq [q],|A|=\Delta\}$ is applied; 
it does not modify the underlying Glauber dynamics itself. 
In other words, the Glauber update rule remains unchanged, 
and adaptively choosing $A$ merely specifies 
which valid grand coupling realization to use at that step. 
Since grand couplings are by definition allowed to be adaptively chosen, this adaptivity does not introduce bias nor affect correctness.
\end{remark}

\subsubsection{Construction of grand coupling and coalescence predicate.}
We now define the global coupling $F(\cdot, U)$ by composing the atomic local updates according to a two-phase schedule.
The resulting operator $F:\Omega \to \Omega$ is represented by the composition list $F^*$, which records the ordered sequence of local updates together with their corresponding random seeds, as described in \Cref{alg:evaluate-F}.
For convenience, we write $F^* \equiv F(\cdot, U)$ to emphasize that the composition list fully specifies the coupling induced by the global seed~$U$.

At the end, if every bounding list $\mathcal{L}(v)$ for $v\in V$ reduces to a singleton, a coalescence is detected and we set $\Phi_F(U)$ to be the corresponding unique configuration.
Otherwise, we set $\Phi_F(U) = \bot$.

\begin{algorithm}[H] 
\caption{Construction of $F(\cdot,U)$ and evaluation of $\Phi_F(U)$} 
\label{alg:construct-F} 

    \KwIn{Graph $G(V,E)$ with maximum degree $\Delta$; $q\ge 2$; global random seed $U\sim\uniform[0,1]$.} 
    \KwOut{Grand coupling $F(\cdot,U)$, represented as a composition list $F^*\equiv F(\cdot,U)$; outcome of coalescence predicate $\Phi_F(U)$.} 

\let\oldnl\nl
\newcommand{\nonl}{\renewcommand{\nl}{\let\nl\oldnl}}
    
    \BlankLine 

    \nonl
    \textbf{Initialization:}
    
        Initialize the global variables $\mathcal{L}(v) \gets [q]$ for all $v \in V$, and $F^*\gets()$ to be an empty list\; 
        $S \gets \textsc{VertexPartition}(G,\Delta)$ \tcp*[r]{Provided by Lemma~\ref{lemma:lll}} 
    
    \BlankLine 
    \nonl
    \textbf{Phase I (Seeding Phase):}
    
    \Indp 
    \label{line:seeding-init-start}\ForEach{$v_i \in S$ in fixed order}{ 
        $\textsc{CleanUp}(v_i,\{v_1,\ldots,v_{i-1}\})$\; 
        $\update(v,\seeding)$\; \label{line:seeding-init-end}
    } 
     \label{line:seeding-drift-start}\For{$t = 1$ \KwTo $T_1 = 5|S|\log |S|$}{ 
        Choose $v \in S$ uniformly at random\; 
        $\textsc{CleanUp}(v,S)$\; 
        $\update(v,\seeding)$ }  \label{line:seeding-drift-end}
    \Indm 
    \BlankLine 

    \nonl
    \textbf{Phase II (Converting and Drifting Phase):}
    
    \Indp 
    \label{line:convert-start}\ForEach{$v_i \in V \setminus S$ in fixed order}{ 
        $\textsc{CleanUp}(v_i,S\cup \{v_1,\ldots,v_{i-1}\})$\; 
        $\update(v,\disjoint)$\;  \label{line:convert-end}
    } 
    \label{line:drift-start}\For{$t = 1$ \KwTo $T_2 = \frac{2(q - \Delta)n \log n}{q - 2.5\Delta}$}{ 
        Choose $v \in V$ uniformly at random\; 
        $\update(v,\disjoint)$\; \label{line:drift-end}
    } 
    \Indm 
    \BlankLine 

    \nonl
    \textbf{Check coalescence:}
    
    \eIf{$|\mathcal{L}(v)| = 1$ for all $v \in V$}{ 
        $\Phi_F(U) \gets$ unique configuration $\conf\in\bigotimes_{v\in V}\mathcal{L}(v)$\; 
    }{ 
        $\Phi_F(U) \gets \bot$\; 
    } 
    \Return $(F^*, \Phi_F(U))$\;
\end{algorithm}

\begin{lemma}[Efficiency and coalescence]
\label{lemma:construct-F}
Suppose $q>(2.5+\eta)\Delta$ with $\eta=2\sqrt{(\log\Delta+1)/\Delta}$. Then \Cref{alg:construct-F} returns a $(F^*,\Phi_F(U))$ in time $\tilde{O}(n\Delta^2)$,
and it holds $\Pr{\Phi_F(U)\neq\bot}\ge 1/2$.
\end{lemma}

The proof of \cref{lemma:construct-F} is deferred to the next subsections.

\begin{lemma}[Evaluation cost]
\label{lemma:evaluate-F}
Given the composition list $F^*$ returned by \Cref{alg:construct-F},
the mapping $F^*(\conf)=F(\conf,U)$ can be evaluated by
Algorithm~\ref{alg:evaluate-F} for any configuration $\conf\in\Omega$ in time $\tilde{O}(n\Delta^2)$.
\end{lemma}

\begin{proof}
By \Cref{lem:compress,lem:seeding,lem:disjoint}, each local update costs 
$O(\Delta\log\Delta)$ in time. Since \cref{alg:construct-F} performs at most 
$O(n\Delta\log n)$ updates across both phases, the total time cost is 
$O(n\Delta^2\log n\log \Delta)$.
\end{proof}

Combining \cref{lemma:construct-F,lemma:evaluate-F} 
with \cref{lemma:blocked-cftp-sampler}, 
we conclude that the perfect sampling algorithm 
terminates in expected time $\tilde{O}(n\Delta^2)$ 
and outputs an exact sample from the uniform distribution $\mu_\Omega$.
This establishes \cref{thm:main}. 

\subsection{Vertex partition via algorithmic LLL}
\label{subsec:vertex-partition}

As a preprocessing step for the global coupling construction in \Cref{alg:construct-F}, we partition the vertex set $V$ into a \emph{seeding set} $S$ and its complement $V \setminus S$.
The goal is to ensure that each vertex has a well-balanced neighborhood across $S$ and $V \setminus S$. 
The seeding set $S$ is used in Phase~I of \Cref{alg:construct-F} to guarantee that the neighborhoods involved in local updates remain sufficiently sparse, which in turn enables efficient list reduction and coalescence guarantees in subsequent phases.


The following lemma shows that such a partition always exists and can be constructed efficiently.

\begin{lemma}
\label{lemma:lll}
    Let $G = (V, E)$ be a graph with maximum degree $\Delta$. 
    Then for $\eta = 2\sqrt{(\log \Delta + 1)/\Delta}$, there exists a subset $S \subseteq V$ such that for every vertex $v \in V$,
    $$
        |\Gamma(v) \cap S| \le \tfrac{1}{2}\Delta, 
        \qquad
        |\Gamma(v) \cap (V \setminus S)| \le \left( \tfrac{1}{2} + \eta \right)\Delta.
    $$
    Moreover, such a set $S$ can be computed in expected time $O(n \Delta)$.
\end{lemma}

\begin{proof}
    We construct $S$ by including each vertex $v \in V$ independently with
    probability $p_0 = \tfrac{1}{2} - \tfrac{\eta}{2}$.
    For each $u \in V$, let $I_u$ be the indicator for the event $u \in S$,
    and for any vertex $v$, define
    $$
        I(v) \defeq \sum\nolimits_{u \in \Gamma(v)} I_u,
        \qquad 
        \E{I(v)} = p_0\, \deg(v) = \left(\tfrac{1}{2} - \tfrac{\eta}{2}\right) \deg(v).
    $$

    We say that a vertex $v$ is \emph{violating} if either  
    $|\Gamma(v)\cap S| > \tfrac{1}{2}\Delta$ or 
    $|\Gamma(v)\cap (V\setminus S)| > (\tfrac{1}{2}+\eta)\Delta$. 
    Both violations correspond to $|I(v) - \E{I(v)}| > \frac{\eta}{2} \Delta$. By a Chernoff bound, for all $v$,
    $$
    \Pr{|I(v) - \E{I(v)}| > \tfrac{\eta}{2}\Delta}
      \le 2 \exp\left( -\frac{(\eta \Delta)^2}{2\deg(v)} \right)
      \le \frac{2}{e^2 \Delta^2}
    $$
    since $\deg(v) \le \Delta$.
    
    Let $p = 2/(e^2 \Delta^2)$ be the failure probability per vertex, and note that
    each bad event (i.e., vertex $v$ violating the bounds) depends only on the
    random choices of vertices in $\Gamma(v)$, and thus at most $d \defeq \Delta$
    other events. The symmetric Lovász Local Lemma (LLL) condition holds because
    $$
        e p (d + 1) \le \frac{2e(\Delta + 1)}{e^2 \Delta^2} < 1.
    $$
    Therefore, with positive probability, no vertex violates the condition, and such a partition exists.

    To obtain such a partition efficiently, we apply the Moser-Tardos
    resampling algorithm~\cite{moser2010constructive}. We associate with each vertex
    $v \in V$ a random bit $s_v \in \{0,1\}$ indicating its membership in $S$
    (initialized to $1$ independently with probability $p_0$), and a counter
    $c_v = |\Gamma(v) \cap S|$.

    We maintain a queue $Q$ of currently violating vertices. Initially, all vertices are placed into $Q$. While $Q$ is nonempty, we:

    \begin{itemize}
        \item Dequeue a vertex $v$ and check whether $c_v$ violates the desired bounds.
        \item If so, we resample the bits $\{s_u : u \in \Gamma(v)\}$.
        \item For each $w \in \Gamma(u)$ where $u \in \Gamma(v)$, update $c_w$, and enqueue $w$ into $Q$.
    \end{itemize}
    
    Each resampling step affects at most $O(\Delta^2)$ entries. By the analysis of
    Moser and Tardos, the expected number of resamplings is $O(n / \Delta)$, so the
    total expected runtime is $O(n\Delta)$.
\end{proof}

\subsection{Phase I: {Seeding}}
\label{subsec:seeding-phase}

Phase~I transforms the initial bounding lists into size-two lists on the seeding set $S$ obtained from \Cref{lemma:lll}. 
It consists of two stages:

\begin{enumerate}
    \item[i.] \emph{Initialization stage:} a deterministic pass over $S$ that ensures $|\mathcal{L}(v)|\in\{2,3\}$ for every $v\in S$;
    \item[ii.] \emph{Randomized drifting stage:} repeated applications of the
    \textsc{Seeding} coupling to uniformly random vertices in $S$, which drive
    all lists in $S$ to satisfy $|\mathcal{L}(v)|=2$ with constant probability.
\end{enumerate}

We first establish the deterministic initialization property and then analyze the convergence of the drifting stage.

\begin{lemma}
\label{lem:seeding-init}
Assume $q \ge (2.5+\eta)\Delta$, and let $S$ be the seeding set constructed in \Cref{lemma:lll}. 
After the initialization stage of Phase~I (Lines~\ref{line:seeding-init-start}--\ref{line:seeding-init-end} in \Cref{alg:construct-F}), 
every vertex $v \in S$ satisfies $|\mathcal{L}(v)|\in\{2,3\}$.
\end{lemma}

\begin{proof}
    We prove by induction on the fixed update order of vertices in $S$.
    Assume that before processing $v_i$, every $v\in P=\{v_1,\ldots,v_{i-1}\}$ satisfies $|\bl(v)|\in \{2,3\}$.
    The procedure $\textsc{CleanUp}(v_i,P)$ chooses a $A\subseteq[q]$ of size $\Delta$ from
    $$
        \slv_P \defeq \bigcup_{u\in \Gamma(v_i)\cup P}\bl(u),
    $$
    greedily: include as many colors as possible from $\slv_P$, and if $|\slv_P|<\Delta$, fill the remaining slots arbitrarily from
    $[q]\setminus \slv_P$.

    Next, $\compress^{A}$ is applied to every $w\in\Gamma(v_i)\setminus P$.
    By \Cref{lem:compress}, after this step each such neighbor satisfies $\bl(w)=A\cup\{c_w\}$ for some $c_w\in[q]\setminus A$.
    Hence, at this moment
    $$
        \slv \defeq \bigcup_{u\in\Gamma(v_i)}\bl(u)
        = (A\cup \slv_P)\cup \{c_w:w\in\Gamma(v_i)\setminus P\}.
    $$
    Therefore,
    $$
        |\slv|
        \le \max(|A|,\,3|\Gamma(v_i)\cap P|)+|\Gamma(v_i)\setminus P|
        \le 2\Delta,
    $$
    where the first inequality uses $|A|=\Delta$ and the inductive assumption $|\bl(v)|\le 3$ for all $v\in P$, and the second inequality follows from $|\Gamma(v_i)|\le\Delta$ and $|\Gamma(v_i)\cap P|\le \tfrac{1}{2}\Delta$ by \Cref{lemma:lll}.

    Thus the precondition of \Cref{lem:seeding}, namely $|\slv|<2\Delta$, holds. Applying the \textsc{Seeding} grand coupling then gives $|\bl(v_i)|\in\{2,3\}$, preserving the induction hypothesis. The claim follows for all $v_i\in S$.
\end{proof}

This completes the deterministic initialization of Phase~I:
after processing all vertices in $S$, each $v \in S$ satisfies
$|\mathcal{L}(v)| \in \{2,3\}$.  
In the subsequent \emph{drifting stage}, we repeatedly apply the
\textsc{Seeding} coupling to uniformly random vertices of $S$,
which further reduces each list to size~$2$ with constant probability.

\begin{lemma}
\label{lem:drifting}
Assume $q \ge (2.5 + \eta)\Delta$, and let $S$ be the seeding set constructed in \Cref{lemma:lll}. 
Let $m = |S|$ and $T_1 = 5 m \log m$. 
Starting from a configuration in which $|\mathcal{L}(v)| \in \{2,3\}$ for all $v \in S$
(as ensured by \Cref{lem:seeding-init}),
after $T_1$ iterations of the drifting step in Phase~I (Lines~\ref{line:seeding-drift-start}--\ref{line:seeding-drift-end} in \Cref{alg:construct-F}), 
we have $|\mathcal{L}(v)| = 2$ for all $v \in S$ with probability at least $3/4$.
\end{lemma}

\begin{proof}
    Let $\mathcal{L}_t$ denote the bounding list configuration after $t$ randomized drifting step, and define
    $$
        V_t \defeq \{v \in S : |\mathcal{L}_t(v)| = 3\},
    $$
    be the set of vertices whose lists have not yet been reduced to size $2$. We introduce the potential function $\Phi_t = |V_t|$ which measures the number of unresolved vertices at time~$t$.

    For each vertex $v\in S$, let
    $$
        \slv_{t}(v) \defeq \bigcup_{u\in\Gamma(v)} \mathcal{L}_t(u)
    $$
    denote the \emph{slack color set} of $v$ at iteration~$t$, i.e., the union of bounding lists over its neighborhood at that time.By property~(\ref{seeding-size}) of \Cref{lem:seeding}, the probability that vertex $v$ remains in the $|\mathcal{L}_t(v)|=3$ state after one update is
    $$
        p_{t,v}
        = \frac{(\,|\slv_{t}(v)| + \Delta - q\,)\,(|\mathcal{S}_t(v)| - 1\,)}{(\,q - \Delta\,)\,\Delta}
        \le\frac{2|\slv_t(v)|}{q-\Delta}-2
    $$
    by assuming $|\slv_t(v)|\le 2\Delta$. The expected change in the potential at each iteration can be expressed as
    $$
        \E{\Phi_{t+1}-\Phi_t \mid V_t}
            \le \frac{1}{m}
                \left(
                    \sum_{v\in S\setminus V_t} p_{t,v}
                    - \sum_{v\in V_t}(1-p_{t,v})
                \right)
            \le \frac{1}{m}\left(
                 2\left(\frac{\sum_{v\in S}|\slv_t(v)|}{q-\Delta} -m\right)- \Phi_t
            \right),
    $$

    Next, we bound the neighborhood slack sizes $|\slv_t(v)|$ that appear in $p_{t,v}$.
    By the greedy construction of the reference set $A$ in \Cref{lem:seeding-init}, the $\textsc{CleanUp}$ operation ensures that
    $$
        |\slv_t(v)|
        \le
        \max\!\bigl(\Delta,\,
            2|\Gamma(v)\cap S|
            + |\Gamma(v)\cap V_t|
        \bigr)
        + |\Gamma(v)\setminus S|\le (1.5+\eta)\Delta+|\Gamma(v)\cap V_t|
    $$
    by applying the degree bounds from \Cref{lemma:lll}, namely $|\Gamma(v)\cap S|\le \tfrac{1}{2}\Delta$ and $|\Gamma(v)\setminus S|\le (\tfrac{1}{2}+\eta)\Delta$.

    Summing over all $v\in S$ yields the aggregate bound
    $$
        \sum_{v\in S} |\slv_t(v)|
            \le (1.5+\eta)m\Delta + \sum_{v\in S}|\Gamma(v)\cap V_t|\le(1.5+\eta)m\Delta+\tfrac{1}{2}\Phi_t\Delta.
    $$

    Plugging into the drift expression and applying $q\ge (2.5+\eta)\Delta$:
    $$
        \E{\Phi_{t+1} - \Phi_t \mid V_t}
        \le \frac{1}{m} \left( 2m\left(\frac{(1.5+\eta)\Delta}{(2.5+\eta)\Delta - \Delta} -1\right) + \frac{\Phi_t \Delta}{q - \Delta} - \Phi_t \right)
        = -\frac{\Phi_t}{m} \left( 1 - \frac{\Delta}{q - \Delta} \right).
    $$

    Let $\delta \defeq \frac{1 - \frac{\Delta}{q - \Delta}}{m}$. Then we have $\E{\Phi_{t+1} \mid \Phi_t} \le \Phi_t (1 - \delta)$,i.e., the potential decreases in expectation by a multiplicative factor. Since $ \delta T = \left(1 - \frac{\Delta}{q - \Delta} \right) \cdot 5\log m \ge \log(4m) $ for all $ q \ge (2.5 + \eta)\Delta $, applying the Multiplicative Drift Theorem and Markov's inequality yields:
    $$
        \Pr{\Phi_T \ge 1} \le \E{\Phi_T} \le m e^{-\delta T} \le \frac{1}{4}.
    $$
    Thus, with probability at least $3/4$, we have $\Phi_T = 0$, meaning all vertices in $S$ have label size 2.
\end{proof}

The two-stage analysis above completes Phase~I:
after the initialization pass (\cref{lem:seeding-init}) and the drifting stage
(\cref{lem:drifting}), all vertices in the seeding set $S$ are $2$-bounded with probability at least $3/4$.
This configuration that is $2$-bounded on seeding vertices serves as the starting point for Phase~II.

\subsection{Phase II: Converting and Drifting}
\label{subsec:phase-II}

Phase~II begins with the $|\bl(v)|=2$ for all vertex in seeding set $S$ by Phase~I and completes the coalescence of the bounding chain.
It consists of two stages:

\begin{enumerate}
    \item[i] \emph{Conversion stage:} a deterministic pass over $V \setminus S$ that ensures $|\mathcal{L}(v)| \in \{1,2\}$ for every $v \in V \setminus S$;
    \item[ii] \emph{Drifting stage:} randomized updates on all vertices
    that drive every list to size~$1$ with constant probability.
\end{enumerate}

We begin with the conversion stage.

\begin{lemma}
\label{lem:conversion}
Assume $q \ge (2.5+\eta)\Delta$, and let $S$ be the seeding set constructed in \Cref{lemma:lll}. 
After the conversion stage of Phase~II (Lines~\ref{line:convert-start}--\ref{line:convert-end} in \Cref{alg:construct-F}), every vertex $v \in V \setminus S$ satisfies $|\mathcal{L}(v)|\in\{1,2\}$.
\end{lemma}

\begin{proof}
    We prove by induction on the fixed update order of vertices in $V\setminus S$.
    Assume that before processing $v_i$, every $v\in P=S\cup\{v_1,\ldots,v_{i-1}\}$
    satisfies $|\bl(v)|\le2$. The procedure $\textsc{CleanUp}(v_i,P)$ constructs a reference set
    $A\subseteq[q]$ of size~$\Delta$ as follows.
    Define
    \[
        \slv_P \defeq \bigcup_{u\in \Gamma(v_i)\cup P}\bl(u),
        \qquad
        \dlv_P \defeq
        \bigcup_{\substack{
            u\in \Gamma(v_i)\cup P:\\
            \forall w\in\Gamma(v_i),\,w\neq u,
            \ \bl(u)\cup\bl(w)=\varnothing
        }}
        \bl(u),
    \]
    where $\slv_P$ collects all colors currently active in the neighborhood of
    $v_i$ and the processed vertices $P$, while $\dlv_P$ collects those colors
    belonging to \emph{disjoint} bounding lists in this neighborhood. \textsc{Greedy} try to build $A$ including as many colors as possible from $\slv_P\setminus\dlv_P$, and then $\dlv_P$.
    
    Let $x$ denote the number of bounding lists $\bl(u)$ for preserved neighbors $u\in P$ that are not contained in $\dlv_P$, and let $y$ be the number of those that are contained in $\dlv_P$.
    Further, let $z = |\Gamma(v_i)\setminus P|$ be the number of unprocessed neighbors of $v_i$. By definition, $x + y + z = |\Gamma(v_i)| \le \Delta$.
    From \Cref{lemma:lll}, we also have $z = |\Gamma(v_i)| - x - y\le \left(\tfrac{1}{2} + \eta\right)\Delta.$

    Next, $\compress^{A}$ is applied to each $w\in\Gamma(v_i)\setminus P$.
    By \Cref{lem:compress}, after this step every such neighbor satisfies
    $\bl(w) = A \cup \{c_w\}$ for some $c_w\in[q]\setminus A$.
    Accordingly, the disjoint and slack color sets at $v_i$ can be written as
    $$
        \dlv
            \defeq
            \bigcup_{\substack{
                u\in \Gamma(v_i):\\
                \forall w\in\Gamma(v_i),\,w\neq u,\
                \bl(u)\cup\bl(w)=\varnothing
            }}
            \bl(u)
            = \dlv_P \setminus A \setminus
                \{\,c_w : w\in\Gamma(v_i)\setminus P\,\},
    $$
    and
    $$
        \slv
            \defeq \bigcup_{u\in\Gamma(v_i)} \bl(u)
            = (A \cup (\slv_P\setminus\dlv_P)\cup \dlv_P)
                \cup (\{\,c_w : w\in\Gamma(v_i)\setminus P\,\}\setminus A\setminus\dlv_P).
    $$

    Let $u =
        \bigl|\{\,c_w : w\in\Gamma(v_i)\setminus P\,\}
            \cup (\dlv_P\setminus A)\bigr|$
    be the number of distinct colors newly introduced outside $A$
    interact to $\dlv$ during \textsc{CleanUp}.
    By definition of $x,y,z$ and the structure of $\slv$ and $\dlv$ above, we have
    $$
        |\slv|
            \le \max(\Delta,\,1.5x + 2y) + z - u,
            \qquad
        |\dlv|
            \ge \max\bigl(0,\,2y - \max(0,\,\Delta - 1.5x)\bigr) - u,
    $$
    since every bounding list $\bl(v)$ not contained in $\dlv$
    contributes at most $1.5$ colors to $\slv$ by \cite{jain2021perfectly}.
    
    Combining the two inequalities yields
    $$
        |\slv| - \tfrac{1}{2}|\dlv|
            \le
            \max(\Delta,\,1.5x + 2y)
            - \max\left(0,\,y - \max(0,\,\tfrac{1}{2}\Delta - 0.75x)\right)+z
            - \tfrac{1}{2}u.
    $$
    We now upper bound the right-hand side of the above inequality under the constraints
    $$
        x+y+z\le \Delta,\qquad
        z\le \bigl(\tfrac{1}{2}+\eta\bigr)\Delta,\qquad
        u\ge 0.
    $$
    Since the expression is decreasing in $u$, it suffices to consider $u=0$.
    For fixed $z$, the feasible region for $(x,y)$ is the triangle
    $\{\,x,y\ge 0,\; x+y\le \Delta-z\,\}$, on which the objective
    $$
    f(x,y,z)
     = \max(\Delta,\,1.5x+2y)
       - \max\left(0,\,y-\max(0,\tfrac{1}{2}\Delta-0.75x)\right)+z
    $$
    is piecewise linear. When $1.5x\ge\Delta$, $f(x,y,z)=1.5x+y+z\le 1.5\Delta$. When $1.5+2y\le\Delta$, $f(x,y,z)=\Delta+z\le(1.5+\eta)\Delta$. Otherwise $f(x,y,z)=2.25x+y+z-0.5\Delta\le 1.5\Delta$.
    
    Substituting this extremal configuration gives
    $$
        |\slv|-\tfrac{1}{2}|\dlv|
        \le (1.5+\eta)\Delta.
    $$
    This establishes the desired bound and verifies that the precondition of \Cref{lem:disjoint} holds for every vertex $v_i$ during the conversion stage.
\end{proof}

Having completed the deterministic conversion stage, we now analyze the subsequent randomized drifting process. 
Starting from the globally $2$-bounded configuration produced above, we repeatedly apply the $\textsc{Disjoint}$ coupling on uniformly random vertices until full coalescence.

\begin{lemma}[{\cite[Theorem~5.1]{jain2021perfectly}}]
\label{lem:disjoint-drifting}
Assume the process starts from a globally $2$-bounded configuration and that $q > (2.5+\eta)\Delta$, where $\eta=2\sqrt{(\log\Delta+1)/\Delta}$. 
Then, after  $T_2=\frac{2(q-\Delta)n\log n}{q-2.5\Delta}$ iterations of the \textsc{Disjoint} coupling, the bounding chain collapses to a singleton with probability at least $3/4$.
\end{lemma}

\begin{remark}
    The proof of \Cref{lem:disjoint-drifting} follows the same potential-based
    argument used in the drifting stage of Phase~I.
    The potential function $\Phi_t=|\{v:|\bl_t(v)|=2\}|$ decreases
    geometrically under the contraction bound provided by
    \Cref{lem:disjoint}, leading to exponential convergence in expectation.
\end{remark}

Combining \Cref{lem:conversion} and \Cref{lem:disjoint-drifting} establishes that, for $q \ge (2.5+\eta)\Delta$, Phase~II succeeds with a  probability of $3/4$, thereby completing the proof of coalescence.

\subsection{Wrapping up}
\label{subsec:put-together}

We now combine the results established in the preceding subsections.  Assume $q \ge (2.5+\eta)\Delta$.
\Cref{lem:seeding-init,lem:drifting} 
show that Phase~I (Seeding phase) produces a configuration that is $2$-bounded over the seeding set $S$ 
with probability at least $3/4$.  
Then, \Cref{lem:conversion,lem:disjoint-drifting} 
ensure that, starting from such a configuration, 
Phase~II (Converting and Drifting phase) reduces all bounding lists to singletons, also with probability at least $3/4$.  

By a union bound, a successful coalescence occurs with probability at least $1/2$. 
Together with \Cref{lemma:evaluate-F}, 
the total expected computation time is $O(n\Delta^2 \log \Delta)$, 
which completes the proof of \Cref{lemma:construct-F}.

Combining \Cref{lemma:construct-F,lemma:blocked-cftp-sampler}, 
we obtain the correctness and efficiency guarantees of the CFTP algorithm, thereby completing the proof of \Cref{thm:main}.

\section{Conclusion}
\label{sec:conclusion}

We developed a general framework for constructing and analyzing grand couplings based on bounding chains, and applied it to the problem of perfect sampling of proper $q$-colorings via the coupling-from-the-past (CFTP) method. This framework allows us to formulate the design of bounding chains as an optimization problem, unifying prior constructions and enabling systematic analysis.

Using this approach, we established that $q = 2.5\Delta$ is the tight threshold for bounding-chain–based CFTP algorithms. Our upper bound is achieved through an explicit grand coupling that minimizes the expected size of bounding lists, while the matching lower bound proves that no contractive bounding chain can succeed below this threshold. Together, these results reveal that the bounding chain paradigm is inherently limited by worst-case local configurations that block contraction.

This insight suggests that surpassing the $2.5\Delta$ barrier will likely require moving beyond the bounding-chain framework itself, or at least beyond worst-case contraction analysis.

\bibliographystyle{alpha}
\bibliography{refs} 
\clearpage


\end{document}